\documentclass[10pt]{amsart}

\usepackage{amsmath,amsfonts,amscd,amssymb,epsf}
\newtheorem{theorem}{Theorem}[section]
\newtheorem{proposition}[theorem]{Proposition}
\newtheorem{lemma}[theorem]{Lemma}

\theoremstyle{definition}

\theoremstyle{remark} 
\numberwithin{equation}{section}

\newcommand{\Z}{{\mathbb{Z}}}

\newcommand{\R}{\mathbb{R}}

\newcommand{\pa}{\partial}

\begin{document}

\title{On the minima and convexity of   Epstein Zeta function}
\author{S.C. Lim$^1$}\email{$^1$sclim@mmu.edu.my}\author{L.P.
Teo$^{2}$}\email{$^2$lpteo@mmu.edu.my}

\keywords{Epstein Zeta function, convexity, minima}
\subjclass[2000]{Primary 11E45, 26B15 } \maketitle

\noindent {\scriptsize \hspace{1cm}$^1$Faculty of Engineering,
Multimedia University, Jalan Multimedia, }

\noindent {\scriptsize \hspace{1.1cm} Cyberjaya, 63100, Selangor
Darul Ehsan, Malaysia.}

\noindent {\scriptsize \hspace{1cm} $^2$Faculty of Information
Technology, Multimedia University, Jalan Multimedia,}

\noindent{\scriptsize \hspace{1.1cm} Cyberjaya, 63100, Selangor
Darul Ehsan, Malaysia.}
\begin{abstract}
Let $Z_n(s; a_1,\ldots, a_n)$ be the Epstein zeta function defined
as the meromorphic continuation of the function
\begin{align*}
 \sum_{k\in\Z^n\setminus\{0\}}\left(\sum_{i=1}^n [a_i
k_i]^2\right)^{-s},\hspace{1cm}\text{Re}\; s>\frac{n}{2}
\end{align*}to the complex plane. We show that for
fixed $s\neq n/2$, the function $Z_n(s; a_1,\ldots, a_n)$, as a
function of $(a_1,\ldots, a_n)\in (\R^+)^n$ with fixed
$\prod_{i=1}^n a_i$, has a unique minimum at the point
$a_1=\ldots=a_n$. When $\sum_{i=1}^n c_i$ is fixed, the function
$$(c_1,\ldots, c_n)\mapsto Z_n\left(s; e^{c_1},\ldots,
e^{c_n}\right)$$ can be shown to be a convex function of any $(n-1)$
of the variables $\{c_1,\ldots,c_n\}$. These results are then
applied to the study of the sign of  $Z_n(s; a_1,\ldots, a_n)$ when
$s$ is in the critical range $(0, n/2)$. It is shown that when
$1\leq n\leq 9$, $Z_n(s; a_1,\ldots, a_n)$ as a function of
$(a_1,\ldots, a_n)\in (\R^+)^n$,
 can be both positive and negative for
every $s\in (0,n/2)$. When $n\geq 10$, there are some open subsets
$I_{n,+}$ of $s\in(0,n/2)$, where $Z_{n}(s; a_1,\ldots, a_n)$ is
positive for all $(a_1,\ldots, a_n)\in(\R^+)^n$. By regarding
$Z_n(s; a_1,\ldots, a_n)$ as a function of $s$, we find that when
$n\geq 10$, the generalized Riemann hypothesis is   false for all
$(a_1,\ldots,a_n)$.
\end{abstract}

\section{Introduction} In \cite{Ep1, Ep2}, Epstein introduced the
following two-dimensional zeta function
\begin{align*}
Z_2(A;s)=\sum_{(j,k)\in
\Z^2\setminus\{(0,0)\}}\left[A(j,k)\right]^{-s},\hspace{1cm}\text{Re}\;
s> 1,
\end{align*}where $$A(j,k)=aj^2+2bjk+ck^2=\begin{pmatrix} j & k\end{pmatrix}\begin{pmatrix} a & b\\b &c \end{pmatrix}
\begin{pmatrix} j\\k\end{pmatrix}$$ is a positive definite real quadratic form associated to the
positive definite symmetric matrix
$$A=\begin{pmatrix} a & b\\b &c \end{pmatrix}.$$Later on,  some generalizations  of
this zeta function to higher dimension were considered \cite{e4, e1,
e5, e2, e6}. One of the generalizations is given by
\begin{align}\label{eq109_3}
Z_n(A;s)=\sum_{k\in
\Z^n\setminus\{0\}}\left[k^tAk\right]^{-s},\hspace{1cm}\text{Re}\;
s> \frac{n}{2},
\end{align}where now $A$ is an $n\times n$ positive definite
symmetric matrix and
$$k^tAk=\sum_{i=1}^n\sum_{j=1}^nA_{ij}k_i k_j$$ is
the associated quadratic form. It was proved that the Epstein zeta
function $Z_{n}(A,s)$ has a meromorphic continuation to the whole
complex plane with a single pole at $s=n/2$, and it satisfies a
functional equation. When $n=1$ and $A=(1)$, $Z_{1}(A;s)$ is nothing
but equal to $2\zeta_R(2s)$, where $\zeta_R(s)$ is  the Riemann zeta
function. The generalized Riemann hypothesis associated to the
Epstein zeta function says that aside from the trivial zeros at
$s=-j, j\in \mathbb{N}$, all the other zeros of $Z_{n}(A;s)$ are
located at the critical line $\text{Re}\;s=\frac{n}{4}$. As far back
as 1947, it has been known \cite{e3} that the Riemann hypothesis is
in general not true for the Epstein zeta function, even for $n=2$.

There exists another interpretation of Epstein zeta function with
more geometric flavour. Let $v_1, \ldots, v_n$ be $n$ linearly
independent vectors in $\R^n$ and let $L$ be the  non-degenerate
lattice generated by $\{v_1, \ldots, v_n\}$, i.e. $$L=\left\{
\sum_{i=1}^n k_i v_i\,:\, (k_1,\ldots, k_n)\in \Z^n\right\}.$$ The
quotient of $\R^n$ by $L$, i.e. $\R^n/L$, is an $n$-dimensional
torus. The function
\begin{align*}
\zeta_n(L;s)=\sum_{v\in L\setminus\{\mathbf{0}\}}\langle v,
v\rangle^{-s},\hspace{1cm}\text{where}\;\;\langle v,w\rangle :=
\sum_{i=1}^nv_iw_i,
\end{align*}is called the zeta function associated to the lattice
$L$.  It is actually the same as the Epstein zeta function
$Z_n(A;s)$ if  $A=T_L^t T_L$, and $T_L$ being the $n\times n$ matrix
whose columns are the generators $v_1, \ldots, v_n$ of the lattice
$L$, i.e.,
$$T_L=\begin{pmatrix} v_1 & \ldots & v_n\end{pmatrix}.$$Given a lattice $L$, its dual lattice $L^*$ is defined as the set
\begin{align*}
L^*=\left\{ w\in\R^n\,:\, \langle w, v\rangle =\sum_{i=1}^n w_i v_i
\in \Z \;\text{for all}\; v\in L\right\}.
\end{align*}$L^*$ is a lattice generated by a dual set of
vectors $\{w_1,\ldots, w_n\}$ with $\langle w_i, v_j\rangle
=\delta_{ij}$. The corresponding matrix $T_{L^*}$ is related to
$T_L$ by $T_{L^*}=(T_L^{-1})^t$. The zeta function $\zeta_n(L^*;s)$
of the dual lattice $L^*$ is related to a  spectral zeta function
associated to the compact torus $\R^n/L$. Given a compact manifold
$M$ with Laplace operator $\Delta_M$, the zeta function of
$\Delta_M$ is defined as
\begin{align}\label{eq109_2}\zeta_{\Delta_M}(s)=\sum_{n=1}^{\infty} \lambda_n^{-s},\end{align} where
$0< \lambda_1\leq \lambda_2\leq \ldots$ are the nonzero eigenvalues
of $\Delta_M$. For the compact torus $\R^n/L$,
 the Laplace operator
$\Delta=\frac{\pa^2}{\pa x_1^2}+\ldots+\frac{\pa^2}{\pa x_n^2}$ on
$\R^n$ descends to the Laplace operator on $\R^n/L$. One can easily
check that the associated zeta function $\zeta_{\Delta_{\R^n/L}}$ is
up to a constant, the zeta function of the dual lattice $L^*$, i.e.
\begin{align*}
\zeta_{\Delta_{\R^n/L}}=c\zeta_{n}(L^*;
s)=cZ_n\left(T_L^{-1}(T_{L}^{-1})^t;s\right).
\end{align*}Under this perspective, the Epstein zeta function enters the realm of
theoretical physics. In quantum theory, very often one needs to
compute the functional determinant of a positive definite
(pseudo)-differential operator $P$ (e.g., the Laplace operator),
which is defined as
\begin{align}\label{eq109_1}
\mathfrak{D}_P=\prod_{i=1}^{\infty}\lambda_i,
\end{align}where $\lambda_i, 1\leq i<\infty$ are the nonzero eigenvalues of
$P$. This functional determinant usually appears as one-loop
partition function of a field theory, and is closely related to
Casimir effect
  as well as the one-loop effective potential of the theory.
In general, the infinite product in \eqref{eq109_1} is divergent and
regularization is required to obtain a finite quantity. One of the
regularization techniques, known as zeta regularization method (see
e.g. \cite{E1, ET,K}), uses the zeta function associated to $P$,
$\zeta_P(s)$, which is defined as in \eqref{eq109_2}. By proving
that $\zeta_P(s)$ has an analytic continuation to a neighborhood of
$s=0$,  the regularized functional determinant $\mathfrak{D}_P$ is
then defined as
\begin{align*}
\mathfrak{D}_P=\exp\left(-\zeta_P'(0)\right).
\end{align*}When the underlying spacetime of the field theory is a
toroidal manifold $T^n\times \R^N$, where $T^n$ is an
$n$-dimensional torus with compactification lengths $L_1, \ldots,
L_n$ and $P$ is the Laplace operator, the associated zeta function
$\zeta_P(s)$ is, up to a constant, equal to
\begin{align*}
\sum_{k\in\Z^n\setminus\{0\}}\left(\sum_{i=1}^n
\left[\frac{k_i}{L_i}\right]^2\right)^{-s+\frac{N}{2}},
\end{align*}which is the Epstein zeta function $Z_n\left(A;
s-\frac{N}{2}\right)$, with $A=\text{diag}\,\{ 1/L_1^2, \ldots,
1/L_n^2\}$ a diagonal matrix. We use the notation
\begin{align}\label{eq109_5} Z_{n}(s; a_1, \ldots,
a_n):=\sum_{k\in\Z^n\setminus\{0\}}\left(\sum_{i=1}^n
[a_ik_i]^2\right)^{-s}
\end{align} to denote this special subclass of the Epstein
zeta function corresponding to $$A=\begin{pmatrix} a_1^2 & 0 &
\ldots & 0\\0 & a_2^2 & \ldots &0\\
\vdots &\vdots & & \vdots\\
0 & 0 &\ldots & a_n^2\end{pmatrix}$$ in \eqref{eq109_3}. Here it is
assumed that $a_1, \ldots, a_n$ are positive real numbers. Another
situation where $Z_n(s;a_1,\ldots, a_n)$  appears is when the
underlying spacetime is a rectangular cavity    $\Omega:=[0,
L_1]\times\ldots\times[0, L_n]\times \R^{N}$. The function $Z_n(s;
a_1,\ldots, a_n)$ satisfies a functional equation (or called
reflection formula), which can be written in the following symmetric
form
\begin{align}\label{eq109_8}
\sqrt{\prod_{i=1}^n a_i}\pi^{-s}\Gamma(s)Z_{n}\left(s; a_1,\ldots,
a_n\right)=\frac{\pi^{s-\frac{n}{2}}}{\sqrt{\prod_{i=1}^n
a_i}}\Gamma\left(\frac{n}{2}-s\right)Z_n\left(\frac{n}{2}-s;
\frac{1}{a_1},\ldots,\frac{1}{a_n}\right).
\end{align}
In physics literature, the Epstein zeta function $Z_{n}(s; A)$
usually appears together with gamma function in the combination
\begin{align}\label{eq109_9} \Gamma(s)Z_{n}(s; A).
\end{align} For example, in \cite{AW}, it is found that the Casimir energy of a massless
scalar field in the rectangular cavity $\Omega=[0,
L_1]\times\ldots\times[0, L_n]\times \R^{N}$ under periodic boundary
conditions is, up to a positive constant, equal to
\begin{align}\label{eq109_7}
-\Gamma\left(-\frac{N}{2}\right)Z_{n}\left(-\frac{N}{2};
\frac{1}{L_1},\ldots, \frac{1}{L_n}\right).
\end{align}
For massless scalar fields under Dirichlet or Neumann boundary
conditions, or electromagnetic fields in cavities with perfectly
conducting walls or walls with infinite permeability, the
corresponding Casimir energy is   expressible as  linear
combinations of expressions of the form \eqref{eq109_7}. Another
example is the the massless scalar field theory with
$\lambda\varphi^4$ interaction on the toroidal manifold $T^n\times
\R^{4-n}$, where $n=1,3,4$. It was found that \cite{c41, c42, c43,
c44, EK1} the topologically generated mass of this theory, to the
one-loop order, is, up to a positive constant, given by :
\begin{align*}
\Gamma\left(\frac{n}{2}-1\right)Z_{n}\left(\frac{n}{2}-1;
\frac{1}{L_1},\ldots,\frac{1}{L_n}\right).
\end{align*} In view of the functional equation \eqref{eq109_8} and the importance of the combination \eqref{eq109_9},
 we define the Xi-function as
\begin{align}\label{eq109_10}
\Xi_n(s; a_1,\ldots, a_n)=\sqrt{\prod_{i=1}^n a_i} \pi^{-s}\Gamma(s)
Z_n(s; a_1,\ldots, a_n).
\end{align}

 In the study of Casimir
effect, the attractive or repulsive nature of the Casimir force, and
the geometric configuration of the spacetime that will minimize or
maximize the Casimir effect are important issues. In the study of
interacting scalar field theory, the sign of the topologically
generated mass   determines whether symmetry breaking mechanism
occurs. Therefore it is important to study the sign of the
Xi-function \eqref{eq109_10} and determine its extrema as a function
of $(a_1,\ldots, a_n)$. This paper is devoted to the study of these
issues. In fact, the main motivation comes from our recent study
  of $\lambda\varphi^4$ interacting fractional Klein--Gordon
field theory on a toroidal manifold $T^n\times \R^N$ \cite{e10}. In
that study, we find that to the one-loop order, the topological
generated mass of a $\alpha$-fractional massless Klein--Gordon field
is given by
\begin{align*}
m_T^2=&\frac{\lambda}{\Gamma(\alpha)}\frac{1}{2^{2\alpha+1}\pi^{2\alpha-\frac{N}{2}}\left[\prod_{i=1}^n
L_i\right]}\Gamma\left(\alpha-\frac{N}{2}\right)Z_{n}\left(\alpha-\frac{N}{2};
\frac{1}{L_1},\ldots,\frac{1}{L_n}\right)\\
=&\frac{\lambda}{\Gamma(\alpha)}\frac{1}{2^{2\alpha+1}\pi^{\frac{n+N}{2}}}\Gamma\left(\frac{n+N}{2}-\alpha\right)
Z_{n}\left( \frac{n+N}{2}-\alpha; L_1, \ldots, L_n\right)\nonumber
\end{align*}when $\alpha\neq N/2$ or $(n+N)/2$. We need to study the
sign of $m_T^2$ as a function of $\alpha, n, N$ and $L_1,\ldots,
L_n$ to determine whether symmetry breaking appears. To solve this
problem, it is necessary to determine the minimum of $\Xi_n(s;
L_1,\ldots, L_n)$ for fixed $\prod_{i=1}^n L_i$ and $s$.

The study of the (local) minima of the general Epstein zeta function
$Z_{n}(A;s)$ \eqref{eq109_3} as a function of $A\in
\{\text{symmetric positive definite} \; n\times n \; \text{real
matrices}\}=\mathcal{P}_n$ with fixed determinant has a long history
[15--38]. Such investigation is important from a number of points of
view. One can read the introduction of the   recent paper \cite{f11}
for an overview of this problem. The local minima of $Z_n(A;s), A\in
\mathcal{P}_n, \det A=1$, for $s$ in some domain of $\R$ has only
been determined for $n=2, 3, 4, 5, 6, 7, 8, 24$; and the problem is
far from being solved. In this paper, we restrict ourselves to a
milder problem of determining the minimum of $Z_{n}(A;s)$ in the
subspace of positive definite symmetric forms consists of diagonal
forms. We give an elementary proof to show that for any $n$ and $s$,
if $\prod_{i=1}^n a_i$ is fixed, the minimum of $\Xi_n(s;
a_1,\ldots, a_n)$ (and hence of $Z_{n}(s; a_1, \ldots, a_n)$)
appears at the point $a_1=\ldots=a_n$. This result implies that if
$\Xi_n (s)=\Xi_n(s; 1,\ldots, 1)\geq 0$, then $\Xi_n(s; a_1, \ldots,
a_n)\geq 0$ for all $(a_1, \ldots, a_n)$. Therefore,   the problem
whether $\Xi_n(s; a_1, \ldots, a_n)$ will be negative for some
$(a_1,\ldots, a_n)$ is reduced to  to the problem of whether
$\Xi_n(s)$ is negative. In section 3, we give a detail study of the
sign of the function $\Xi_n(s)$. We show that for $1\leq n\leq 9$,
$\Xi_n(s)$ is negative for all $s\in (0, n/2)$; whereas for $n\geq
10$, there is a nonempty subset of $(0, n/2)$ where $\Xi_n(s)$ is
positive. In section 4, we prove that for fixed $s$, when
$\sum_{i=1}^n c_i$ is fixed, the function
\begin{align}\label{eq1017_8}(c_1,\ldots, c_n)\mapsto Z_n\left(s; e^{c_1},\ldots,
e^{c_n}\right)\end{align} is a convex function of any $(n-1)$ of the
variables $\{c_1,\ldots,c_n\}$. In fact, we prove a stronger result.
We show that for a function of the form
$$\mathcal{F}_n(x_1,\ldots, x_n)=\prod_{i=1}^n f(x_i),$$ where
$f:\R\rightarrow \R^+$ is a positive twice continuously
differentiable function, $\mathcal{F}_n$ is convex   if the function
$\log f(x)$ is convex. Finally, the result about the convexity of
the function \eqref{eq1017_8} can be used to obtain conclusion on
the connectivity and convexity of some regions of $(c_1,\ldots,
c_n)$ where $Z_n(s; e^{c_1},\ldots, e^{c_n})$ is negative.

\section{The minimum of Epstein zeta function}
In this section, we show that for $s\notin \{ n/2\}\cup\{ 0, -1, -2,
\ldots\}$, the function $Z_n(s; a_1,\ldots, a_n), (a_1, \ldots,
a_n)\in (\R^+)^n$ has a minimum at $a_1=\ldots=a_n$ when
$\prod_{i=1}^n a_i$ is fixed.

There are various ways to obtain the meromorphic continuation of the
Epstein zeta function $Z_n (s; a_1, \ldots, a_n)$ \eqref{eq109_5}.
Two of them will be useful to us. The first one has the form (see
e.g. \cite{f11}):
\begin{align}\label{eq1011_1}
&V\pi^{-s}\Gamma(s)Z_n(s; a_1,\ldots, a_n)=\Xi_n(s; a_1,\ldots,
a_n)\\\nonumber=&-\frac{V}{s}
-\frac{V^{-1}}{\left(\frac{n}{2}-s\right)} + V\int_1^{\infty}
t^{s-1}\sum_{k\in \Z^n\setminus\{0\}}\exp\left(-\pi t\sum_{i=1}^n
[a_i k_i]^2\right)dt\\&+V^{-1}\int_1^{\infty}
t^{\frac{n}{2}-s-1}\sum_{k\in \Z^n\setminus\{0\}}\exp\left(-\pi
t\sum_{i=1}^n \left[\frac{k_i}{a_i}\right]^2\right)dt,\nonumber
\end{align}
where $$V=\sqrt{\prod_{i=1}^n a_i}.$$This formula shows clearly the
Xi-function $\Xi_n(s;a_1,\ldots, a_n)$ only has simple poles at
$s=0$ and $s=\frac{n}{2}$. Moreover, since the gamma function
$\Gamma(s)$ has simple poles at $s=0,-1, -2, -3,\ldots$, the Epstein
zeta function $Z_n(s; a_1, \ldots, a_n)$ is identically zero at the
negative integer points $-1, -2, -3,\ldots$. The second formula is
one form of the Chowla-Selberg formula \cite{d19,
d20}:\begin{align}\label{eq830_9}
&\pi^{-s}\Gamma(s)Z_{n}(s; a_1,\ldots, a_n)\\
=\nonumber &2a_1^{-2s}\pi^{-s}\Gamma(s)\zeta_R(2s)
+2\sum_{j=1}^{n-1}\frac{\pi^{-s+\frac{j}{2}}\Gamma\left(s-\frac{j}{2}\right)}
{a_{j+1}^{2s-j}\prod_{l=1}^{j}a_l} \zeta_R(2s-j) +4 \sum_{j=1}^{n-1}
\frac{1}{\prod_{l=1}^{j}a_l}\times\\\nonumber&\sum_{\mathbf{k}\in
\Z^j\setminus\{\mathbf{0}\}}\sum_{p=1}^{\infty}\frac{1}{(pa_{j+1})^{s-\frac{j}{2}}}\left(\sum_{l=1}^j
\left[\frac{k_l}{a_l}\right]^2\right)^{\frac{s}{2}-\frac{j}{4}}K_{s-\frac{j}{2}}\left(
2\pi p a_{j+1}\sqrt{\sum_{l=1}^j
\left[\frac{k_l}{a_l}\right]^2}\right),
\end{align}where $\zeta_R(s)=\sum_{k=1}^{\infty} k^{-s}$ is the
Riemann zeta function.
 From \eqref{eq830_9}, it is easy to
verify that for any $\lambda\in\R^+$,
\begin{align*}
Z_n(s; \lambda a_1,\ldots, \lambda a_n  )=\lambda^{-2s} Z_n(s; a_1,
\ldots, a_n).
\end{align*}Therefore, it only makes sense to look for the minimum of
$Z_n(s;   a_1,\ldots,  a_n  )$ as a function of $(a_1,\ldots, a_n)$
when $V$ is fixed. Without loss of generality, it suffices to
consider $V=1$. We are going to make use of \eqref{eq1011_1}  to
find the minimum of the Xi-function $\Xi_n(s; a_1, \ldots, a_n)$ as
a function of $(a_1,\ldots, a_n)$  for fixed $s\neq 0, n/2$ and when
$V=1$.  Since multiplying a constant does not affect the minimum of
a function, this will give the minimum of the Epstein zeta function
$Z_n(s; a_1,\ldots, a_n )$ as a function of $(a_1,\ldots, a_n)$ for
$s\notin \{n/2\}\cup\{0,-1,-2, \ldots\}$ and when $V=1$. By
\eqref{eq1011_1}, for fixed $s$ and $V=1$, the minimum of $\Xi_n(s;
a_1, \ldots, a_n)$ is the same as the minimum of
\begin{align}\label{eq1011_2}
\Lambda_n(s;a_1,\ldots, a_n)=&\int_1^{\infty} t^{s-1}\sum_{k\in
\Z^n\setminus\{0\}}\exp\left(-\pi t\sum_{i=1}^n [a_i
k_i]^2\right)dt\\&+\int_1^{\infty} t^{\frac{n}{2}-s-1}\sum_{k\in
\Z^n\setminus\{0\}}\exp\left(-\pi t\sum_{i=1}^n
\left[\frac{k_i}{a_i}\right]^2\right)dt.\nonumber
\end{align}Define the theta function $\vartheta(t)$ by
\begin{align}\label{eq1023_1}
\vartheta(t)=\vartheta_3(0, e^{-\pi t})=\sum_{k=-\infty}^{\infty}
e^{-\pi t k^2}
\end{align}Here $\vartheta_3(z,q)$ is a Jacobi theta function
\cite{g1}. The theta function $\vartheta(t)$ satisfies a reflection
formula: \begin{align}\label{eq1011_6}
\vartheta(t)=\frac{1}{\sqrt{t}}\vartheta\left(\frac{1}{t}\right).
\end{align}The theta function \eqref{eq1023_1} can be used to reexpress \eqref{eq1011_2}  as
\begin{align}\label{eq1011_2_1}
\Lambda_n(s;a_1,\ldots, a_n)=&\int_1^{\infty}
t^{s-1}\left(\prod_{i=1}^n
\vartheta\left(ta_i^2\right)-1\right)dt\\&+\int_1^{\infty}
t^{\frac{n}{2}-s-1}\left(\prod_{i=1}^n
\vartheta\left(\frac{t}{a_i^2}\right)-1\right)dt.\nonumber
\end{align}
To determine the minimum of $\Lambda_n(s; a_1,\ldots, a_n)$, we need
the following:
\begin{proposition}\label{p1}
The function $\log\vartheta(e^u)$, $u\in\R$ is strictly convex.
Namely, for any $m$ distinct points  $u_1, \ldots, u_m\in\R$ and $m$
constants $\lambda_1,\ldots, \lambda_m\in (0,1)$ such that
$\sum_{i=1}^m \lambda_i=1$, we have
\begin{align*}
\log\vartheta \left(\exp\left[\sum_{i=1}^m
\lambda_iu_i\right]\right)< \sum_{i=1}^m\lambda_i
\log\vartheta(e^{u_i}).
\end{align*}
\end{proposition}
This proposition can be used to prove the main result of this
section.
\begin{theorem}\label{th1}
For fixed $s\in\R\setminus\{0, n/2\}$, the function $\Xi_n(s;
a_1,\ldots, a_n)$, where $(a_1,\ldots, a_n)\in (\R^+)^n$ with
$\prod_{i=1}^n a_i=1$, has a unique minimum at $a_1=\ldots=a_n=1$.
Therefore for fixed $s\in\R\setminus\left(\{n/2\}\cup\{ 0, -1, -2,
-3, \ldots\}\right)$, the same statement holds for the function
$Z_{n}\left(s; a_1,\ldots, a_n\right)$.
\end{theorem}First we show how we prove this theorem from
Proposition \ref{p1}.
\begin{proof}
By the strict convexity of $\log \vartheta(e^u)$ asserted in
Proposition \ref{p1}, we find that if $(a_1,\ldots, a_n)\neq
(1,\ldots,1)$, then
\begin{align*}
&\log\prod_{i=1}^n \vartheta\left(a_i^2 t\right)= \sum_{i=1}^n \log
\vartheta\left(\exp\left[\log t +\log a_i^2\right]\right) \\>& n
\log\vartheta\left(\exp\left\{\frac{1}{n} \sum_{i=1}^n\left[\log
t+\log a_i^2\right]\right\}\right)= \log \vartheta(t)^n.
\end{align*}Therefore,
\begin{align}\label{eq1011_3}
\prod_{i=1}^n \vartheta\left(a_i^2 t\right)> \vartheta(t)^n.
\end{align}
Similarly,\begin{align}\label{eq1011_4} \prod_{i=1}^n
\vartheta\left(\frac{ t}{a_i^2}\right)> \vartheta(t)^n.
\end{align}Since $t^{\beta-1}>0$ for all $t>0$ and $\beta\in\R$ and
\begin{align*}
\prod_{i=1}^n \vartheta\left(b_i^2 t\right)-1> 0
\hspace{0.5cm}\text{for all}\; (b_1,\ldots, b_n)\in (\R^+)^n,
\end{align*}we conclude from \eqref{eq1011_2_1}, \eqref{eq1011_3} and \eqref{eq1011_4}
that for all $(a_1,\ldots, a_n)\in (\R^+)^n\setminus\{(1,\ldots,
1)\}$ with $\prod_{i=1}^n a_i=1$, the inequality
\begin{align*}
\Lambda_n(s; a_1, \ldots, a_n)> \Lambda_n(s; 1,\ldots, 1)
\end{align*}holds. The  assertion of the theorem follows.
\end{proof}
Now we return to the  proof of Proposition \ref{p1}.
\begin{proof}Let
\begin{align}\label{eq1011_5}
g(u)&=\log\vartheta(u), \hspace{2.3cm} u\in \R^+,
\\f(u)&=g(e^u)=\log\vartheta(e^u),\hspace{1cm}u\in \R.\nonumber
\end{align}
We need to show that $f$ is strictly convex. Since $f$ is an
infinitely differentiable function, we need to show that
$f^{\prime\prime}(u)>0$ for all $u\in\R$. First, note that the
reflection formula \eqref{eq1011_6} implies that
\begin{align}\label{eq1011_7}
f(-u)=\frac{u}{2}+f(u).
\end{align} Therefore
$$f^{\prime\prime}(-u)=f^{\prime\prime}(u).$$ As a result, we only need  to show that $f^{\prime\prime}(u)> 0$
for all $u\in\R^+$. By definition,
\begin{align*}
g'(u)=\frac{\vartheta'(u)}{\vartheta(u)},
\hspace{1cm}g^{\prime\prime}(u)=\frac{\vartheta^{\prime\prime}(u)\vartheta(u)-\vartheta'(u)^2}{\vartheta
(u)^2},
\end{align*}and\begin{align*}
f^{\prime\prime}(u) =& e^{2u}g^{\prime\prime}(u)+e^ug'(e^u)\\
=&e^{2u}\frac{
\vartheta^{\prime\prime}(e^u)\vartheta(e^u)-\vartheta'(e^u)^2+e^{-u}\vartheta(e^u)\vartheta'(e^u)}{\vartheta(e^u)^2}.
\end{align*}Hence, we have to prove that for all $v\geq 1$,
\begin{align*}h(v)=&\vartheta^{\prime\prime}(v)\vartheta(v)-\vartheta'(v)^2+v^{-1}\vartheta(v)\vartheta'(v)>0.
\end{align*}Again by definition,\begin{align}\label{eq1011_9}h(v)
=&\sum_{j\in\Z} \pi^2j^4 e^{-\pi v j^2}\sum_{k\in\Z}e^{-\pi v
k^2}-\sum_{j\in \Z}\pi j^2 e^{-\pi vj^2}\sum_{k\in\Z} \pi k^2e^{-\pi
vk^2} -v^{-1}\sum_{j\in\Z} e^{-\pi v j^2}\sum_{k\in\Z}\pi k^2e^{-\pi
vk^2}\\
=&\frac{\pi}{2}\sum_{(j,k)\in\Z^2}\left[\pi j^4+\pi k^4-2\pi
j^2k^2-\frac{1}{v}(j^2+k^2)\right] e^{-\pi v (j^2+k^2)}.\nonumber
\end{align}For fixed $v\geq 1$ and $(j,k)\in \Z^2$, define
\begin{align*}
C(j,k)  =\pi j^4+\pi k^4-2\pi
j^2k^2-\frac{1}{v}(j^2+k^2)=\pi(k^2-j^2)^2-\frac{1}{v}(j^2+k^2).
\end{align*} If $C(j,k)>0$ for all $(j,k)\in\Z^2$, it follows directly that $h(v)>0$. However, $C(k,k)=
-2v^{-1}k^2\leq 0$, which renders the verification of  $h(v)>0$
slightly more complicated. We need to compare the values of $C(j,k)$
for different pairs of $(j,k)$. Since
$C(j,k)=C(-j,k)=C(j,-k)=C(-j,-k)$ and $C(j,k)=C(k,j)$, we need only
to consider $C(j,k)$ as a function of $(j,k)\in \mathbb{N}_0$ with
$j\leq k$. Here $\mathbb{N}_0=\mathbb{N}\cup\{0\}$. We claim that
\begin{align*}
&(a) \; \text{If}\; j\neq k, C(j,k)>0,\\
&(b)\; \text{For any}\; k\geq 1, C(k-1,
k)+\frac{1}{2}C(k,k)>0.\hspace{6cm}
\end{align*} By rewriting $C(j,k)$ as
\begin{align*}
C(j,k)=\pi j^4-\left(2\pi k^2+\frac{1}{v}\right)j^2+\pi
k^4-\frac{k^2}{v},
\end{align*}  we  see that for fixed $k$, $C(j,k)$ is decreasing if
$j\in [0, k]$. Therefore, to prove $(a)$, it suffices to verify that
$C(k-1, k)>0$. Since $C(k,k)\leq 0$, this will follow if we prove
$(b)$. Now for $v\geq 1$, the function
\begin{align*}
q(k)=C(k-1,k)+\frac{1}{2}C(k,k)
=&\pi(4k^2-4k+1)-\frac{1}{v}(2k^2-2k+1)-\frac{k^2}{v}\\
\geq &\pi(4k^2-4k+1)-(3k^2-2k+1)\\
=&(4\pi -3)k^2 -2(2\pi-1)k +\pi-1
\end{align*}is  increasing for $k\geq 1$, and $q(1)= \pi-2>0$,  this
proves $(b)$. Returning to the proof of $h(v)>0$, we write the
summation over $(j,k)\in\Z^2$ in \eqref{eq1011_9} as
\begin{align*}
\frac{\pi}{2}\Biggl\{\sum_{\substack{(j,k)\in \Z^2\\|j-k|\geq
2}}+\sum_{\substack{(j,k)\in \Z^2\\|j-k|=
1}}+\sum_{\substack{(j,k)\in \Z^2\\j=k}}\Biggr\}C(j,k)e^{-\pi v
(j^2+k^2)}.
\end{align*}By $(a)$, the first sum is strictly positive. Now since $C(0,0)=0$ and
$C(\pm j,\pm k)=C(j,k)=C(k,j)$, the last two sums can be written as
\begin{align*}
2\pi\sum_{k=1}^{\infty}C(k-1, k) e^{-\pi((k-1)^2+k^2)}
+\pi\sum_{k=1}^{\infty}C(k,k)e^{-2\pi k^2}.
\end{align*}Using the fact that $e^{-x}$ is a decreasing function,
we find that this term is larger than
\begin{align*}
2\pi\sum_{k=1}^{\infty}\left(C(k-1,1)+\frac{1}{2}C(k,k)\right)
e^{-2\pi k^2},
\end{align*}which, by $(b)$, is positive. This concludes the proof
that $h(v)>0$, and therefore the assertion of the proposition
follows.
\end{proof}

\section{The sign of Epstein zeta function} In this section,
we are going to study the sign of the Xi-function $\Xi_n(s;
a_1,\ldots, a_n)$ when $s\in\R \setminus \{0, n/2\}$. First, since
$\pi^{-s}\Gamma(s)>0$ for $s>0$ and it is obvious from the power
series definition of $Z_{n}(s; a_1, \ldots, a_n)$ \eqref{eq109_5}
that $Z_n(s; a_1,\ldots, a_n)>0$ for all $s>n/2$, we immediately
deduce that
$$\Xi_n(s; a_1, \ldots, a_n)>0 \hspace{1cm}\text{for all}\;\;
s>\frac{n}{2} \;\;\text{and}\;\;(a_1,\ldots, a_n)\in (\R^+)^n.$$
From the functional equation \eqref{eq109_8}, we then obtain
$$\Xi_n(s; a_1, \ldots, a_n)>0 \hspace{1cm}\text{for all}\;\;
s<0 \;\;\text{and}\;\;(a_1,\ldots, a_n)\in (\R^+)^n.$$In other
words, $\Xi_n(s; a_1, \ldots, a_n)>0 $ for all $s\in (-\infty,
0)\cup(n/2, \infty)$. For the remaining case where $s\in (0,n/2)$,
the sign of the function $\Xi_n(s; a_1,\ldots, a_n)$ is not trivial.
Using the formula \eqref{eq1011_1}, we find that
 around the two simple poles $s=0$ and $s=\frac{n}{2}$,
\begin{align*}
\Xi_n(s; a_1,\ldots, a_n)=&-\frac{V}{s}+O(1)\hspace{1cm}\text{as}\;\; s\rightarrow 0,\\
\Xi_n(s; a_1, \ldots,
a_n)=&-\frac{V^{-1}}{\frac{n}{2}-s}+O(1)\hspace{1cm}\text{as}\;\;
s\rightarrow \frac{n}{2}.
\end{align*}Therefore we can conclude that\begin{proposition}\label{p2} For any $a=(a_1,\ldots, a_n)\in(\R^+)^n$, one
can find a right nonempty neighbourhood  $I_1(a)=(0, s_{1}(a))$ of
$s=0$ and a left nonempty neighbourhood $I_2(a)=\left(s_{2}(a),
\frac{n}{2}\right)$ of $n/2$ such that
\begin{align*}
\Xi_n(s; a_1,\ldots, a_n)<0, \hspace{1cm}\text{for all}\;\; s\in
I_1(a)\cup I_2(a).
\end{align*}\end{proposition}In fact, when $n=1$, using the fact that
$$\Xi_1(s;a)=\sqrt{a}
\pi^{-s}\Gamma(s)Z_1(s;a)=2a^{-2s+\frac{1}{2}}\pi^{-s}\Gamma(s)\zeta_R(2s),$$
 it is easy to verify a stronger result:
\begin{proposition}
$$\Xi_1(s;a)<0,\hspace{1cm}\text{for all}\;\; s\in \left(0,
\frac{1}{2}\right).$$
\end{proposition}\begin{proof}
It suffices to show that $\zeta_R(s)<0$ for all $s\in (0,1)$. For
$s>0$, there is a well-known formula
\begin{align*}
\sum_{k=1}^{\infty}\frac{(-1)^{k-1}}{k^s}=\left(1-2^{1-s}\right)\zeta_R(s).
\end{align*}The alternating series on the left hand side is
convergent and positive for all $s>0$. However, for $s\in (0,1)$,
$2^{1-s}>1$. Therefore, $\zeta_R(s)<0$ for all $s\in (0,1)$. This
completes the proof.
\end{proof}In the study of the zeros of Riemann zeta function, this result asserts that $\zeta_R(s)$ has no nontrivial
zeros in $(0, 1)$, a necessary condition for the validity of Riemann
hypothesis. For $n\geq 2$, in view of Proposition \ref{p2}, we find
that a necessary condition for the validity of the generalized
Riemann hypothesis  for $Z_n(s; a_1,\ldots, a_n)$ at a fixed
$(a_1,\ldots, a_n)\in (\R^+)^n$,  is that  $\Xi_n(s; a_1,\ldots,
a_n)<0$ for all $s\in (0, n/2)$. However, this is not true for all
$(a_1,\ldots, a_n)\in (\R^+)^n$. In fact, we can show that
\begin{proposition}\label{p4}
Given $n\geq 2$, $j\in \{1,\ldots, n\}$ and $s\in (0,n/2)$, if we
fix $a_1,\ldots, a_{j-1}$, $a_{j+1},\ldots, a_n$ and vary $a_j$,
then as $a_j$ is large enough, $\Xi_n(s; a_1,\ldots, a_n)$ is
positive.
\end{proposition}
\begin{proof}Since $\Xi_n (s; a_1,\ldots, a_n)$ is symmetric in the
variables $a_1,\ldots, a_n$, it is sufficient to show that given
$n\geq 2$ and $s\in (0,n/2)$, there exists $j\in \{1,\ldots, n\}$
such that if we fix $a_1,\ldots, a_{j-1}$, $a_{j+1},\ldots, a_n$ and
vary $a_j$, then as $a_j$ is large enough, $\Xi_n(s; a_1,\ldots,
a_n)$ is positive.

 First we consider the case $s=\frac{m}{2}$, where $m\in
\{1,2,\ldots, n-1\}$. Equation \eqref{eq830_9} and the fact that
(see \cite{e10})
\begin{align*}\pi^{-s}\Gamma\left(s\right)\zeta_R(2s)
=& -\frac{1}{2}\left\{\frac{1}{s}+\log
(4\pi)+\psi(1)\right\}+O(s),\hspace{1cm}\text{as}\;\; s\rightarrow
0,\\\pi^{-s}\Gamma\left(s\right)\zeta_R(2s) =&
\frac{1}{2}\left\{\frac{1}{s-\frac{1}{2}}-\log
(4\pi)-\psi(1)\right\}+O\left(s-\frac{1}{2}\right),\hspace{1cm}\text{as}\;\;
s\rightarrow \frac{1}{2},\nonumber
\end{align*}
give
\begin{align}\label{eq1023_2}
&V^{-1}\Xi_n\left(\frac{m}{2}; a_1,\ldots,
a_n\right)=\frac{1}{\prod_{l=1}^m
a_l}\left\{\log\frac{a_{m+1}}{a_m}-\log(4\pi)-\psi(1)\right\}\\
+\nonumber &2\sum_{j\in\{0,1,\ldots,n-1\}\setminus\{m-1,
m\}}\frac{\pi^\frac{j-m}{2}\Gamma\left(\frac{m-j}{2}\right)}
{a_{j+1}^{m-j}\prod_{l=1}^{j}a_l} \zeta_R(m-j) +4 \sum_{j=1}^{n-1}
\frac{1}{\prod_{l=1}^{j}a_l}\times\nonumber\\\nonumber&\sum_{\mathbf{k}\in
\Z^j\setminus\{\mathbf{0}\}}\sum_{p=1}^{\infty}\frac{1}{(pa_{j+1})^{\frac{m-j}{2}}}\left(\sum_{l=1}^j
\left[\frac{k_l}{a_l}\right]^2\right)^{\frac{m-j}{4}}K_{\frac{m-j}{2}}\left(
2\pi p a_{j+1}\sqrt{\sum_{l=1}^j
\left[\frac{k_l}{a_l}\right]^2}\right).
\end{align}Since the function $\pi^{-s}\Gamma(s)\zeta_R(2s)$ is
strictly positive for $s<0$ or $s>\frac{1}{2}$ and the function
$K_s(z)$ is positive for all $s\in\R$ and $z\in \R^+$, the last two
terms of \eqref{eq1023_2} are always positive. The sign of the first
term depend on the ratio $a_{m+1}/a_m$. It is easy to see that   if
$a_1,\ldots, a_m, a_{m+2},\ldots, a_n$ are kept fixed, then for
$a_{m+1}$ large, the first term is positive and therefore
$\Xi_{n}\left(\frac{m}{2}; a_1, \ldots, a_n\right)>0$. This
establishes the proposition when $s\in \{m/2\,:\, m=1,2,\ldots,
n-1\}$.

For the general case,  given $s\in (0, n/2)\setminus\{m/2\,:\,
m=1,2,\ldots, n-1\}$,  there is a unique $m\in \{0, \ldots, n-1\}$
such that $0<2s-m<1$. Equation \eqref{eq830_9} gives
\begin{align*}
&V^{-1}\Xi_n(s; a_1,\ldots,
a_n)=\frac{\pi^{\frac{m}{2}-s}\Gamma\left(s-\frac{m}{2}\right)}
{a_{m+1}^{2s-m}\prod_{l=1}^{m}a_l} \zeta_R(2s-m)\\
\nonumber
&+2\sum_{j=0}^{m-1}\frac{\pi^{\frac{j}{2}-s}\Gamma\left(s-\frac{j}{2}\right)}
{a_{j+1}^{2s-j}\prod_{l=1}^{j}a_l} \zeta_R(2s-j)
+2\sum_{j=m+1}^{n-1}\frac{\pi^{\frac{j}{2}-s}\Gamma\left(s-\frac{j}{2}\right)}
{a_{j+1}^{2s-j}\prod_{l=1}^{j}a_l} \zeta_R(2s-j) +4 \sum_{j=1}^{n-1}
\frac{1}{\prod_{l=1}^{j}a_l}\times\\\nonumber&\sum_{\mathbf{k}\in
\Z^j\setminus\{\mathbf{0}\}}\sum_{p=1}^{\infty}\frac{1}{(pa_{j+1})^{s-\frac{j}{2}}}\left(\sum_{l=1}^j
\left[\frac{k_l}{a_l}\right]^2\right)^{\frac{s}{2}-\frac{j}{4}}K_{s-\frac{j}{2}}\left(
2\pi p a_{j+1}\sqrt{\sum_{l=1}^j
\left[\frac{k_l}{a_l}\right]^2}\right).
\end{align*}Only the first term is negative. If $a_1,\ldots, a_{m},
a_{m+2},\ldots, a_n$ are fixed and $a_{m+1}\rightarrow \infty$, the
second term is fixed but the first term approaches zero. Therefore,
 $\Xi_n(s; a_1,\ldots,
a_n)>0$ when $a_{m+1}$ is large enough, and our assertion is proved.

We would like to remark that the case  $s=\frac{m}{2}$, where $m\in
\{1,2,\ldots, n-1\}$ can actually be deduced from the general case
$s\in (0, n/2)\setminus\{m/2\,:\, m=1,2,\ldots, n-1\}$ using the
joint continuity of $\Xi_n(s; a_1, \ldots, a_n)$ as a function of
$s$ and $(a_1,\ldots, a_n)$.
\end{proof} This proposition shows that if $s\in (0, n/2)$, then when one of the ratios between the
$a_i$'s, $1\leq i\leq n$ is large, the Epstein zeta function
$Z_{n}(s; a_1,\ldots, a_n)$ is positive, and therefore the
generalized Riemann hypothesis does not hold for this zeta function.
It is then interesting to ask whether the Riemann hypothesis will
still be valid for some $(a_1,\ldots, a_n)\in (\R^+)^n$. Another
interesting question which is related to our problem \cite{e10} of
determining whether there exists symmetry breaking mechanisms in an
interacting fractional Klein--Gordon field is, for what values of
$s$, the Xi-function $\Xi_n(s; a_1,\ldots, a_n)$ will be negative
for some $(a_1,\ldots, a_n)$. Recall that we have proved in section
2 that for any fixed $s$, the minimum of $Z_{n}(s; a_1,\ldots,
a_n)$, as a function of $(a_1,\ldots, a_n)\in (\R^+)^n$ with
$\prod_{i=1} a_i$ fixed, appears at the point $a_1=\ldots=a_n$.
Therefore to answer   these two questions, we need  to  study  the
function $\Xi_n(s)=\Xi_n(s; 1,\ldots, 1)$ first.

From \eqref{eq1011_1}, we find that
\begin{align}\label{eq1015_3}
\Xi_n(s)=-\frac{1}{s} -\frac{1}{\frac{n}{2}-s}+\int_1^{\infty}
t^{s-1}\left(\vartheta(t)^n-1\right)dt+\int_1^{\infty}
t^{\frac{n}{2}-1-s}\left(\vartheta(t)^n-1\right)dt.
\end{align}
In order to have a more unified treatment for all $n$, we define
$\hat{\Xi}_n(s)=\Xi_n\left(\frac{ns}{2}\right)$. Then
 \begin{align}\label{eq1015_10}
\hat{\Xi}_n(s)=-\frac{2}{n}\left(\frac{1}{s}
+\frac{1}{1-s}\right)+\int_1^{\infty}
t^{\frac{ns}{2}-1}\left(\vartheta(t)^n-1\right)dt+\int_1^{\infty}
t^{\frac{n}{2}(1-s)-1}\left(\vartheta(t)^n-1\right)dt.
\end{align} This formula shows that for all $s\in
(0,1)$ and $n\in \mathbb{N}$,
\begin{align}\label{eq1015_1}\hat{\Xi}_{n+1}(s)
>\hat{\Xi}_{n}(s).\end{align} Moreover, for any fixed $s$, since
\begin{align*}
\frac{2}{n}\left(\frac{1}{s}+\frac{1}{1-s}\right) \longrightarrow 0
\hspace{1cm}\text{as}\;\; n\rightarrow \infty,
\end{align*}and
\begin{align*}
&\int_1^{\infty} t^{\frac{ns}{2}-1}
\left(\vartheta(t)^n-1\right)dt+\int_1^{\infty}t^{\frac{n}{2}(1-s)-1}\left(\vartheta(t)^n-1\right)dt
\\\geq & \int_1^{\infty} t^{\frac{s}{2}-1}
\left(\vartheta(t)^n-1\right)dt+\int_1^{\infty}t^{\frac{1}{2}(1-s)-1}\left(\vartheta(t)^n-1\right)dt>0,
\end{align*}therefore $\hat{\Xi}_n(s)$ becomes positive when $n$ is
large enough. In fact for $s=1/2$, a lengthy analysis (31 pages) has
been done in \cite{g2} and it was proved that\begin{lemma}\label{l1}
\begin{align*}\begin{cases}\Xi_n\left(\frac{n}{4}\right)=\hat{\Xi}_n\left(\frac{1}{2}\right)<0,
\hspace{1cm}\text{if}\;\;& 1\leq n\leq 9;\\
\Xi_n\left(\frac{n}{4}\right)=\hat{\Xi}_n\left(\frac{1}{2}\right)>0,
\hspace{1cm}\text{if}\;\; & n\geq 10.
\end{cases}
\end{align*}\end{lemma}\begin{proof}Here we provide a much shorter proof of this result.
Using the inequality \eqref{eq1015_1}, we need
only to show that
\begin{align}\label{eq1015_2}
\Xi_9\left(\frac{9}{4}\right)<0
\hspace{1cm}\text{and}\hspace{1cm}\Xi_{10}\left(\frac{5}{2}\right)>0.
\end{align}Using the incomplete gamma function
\begin{align*}
\Gamma(\beta,x)=\int_x^{\infty} t^{\beta-1} e^{-t} dt,
\hspace{1cm}x>0,
\end{align*}we can rewrite  \eqref{eq1015_3} as
\begin{align}\label{eq1015_4}
\Xi_n(s)=&-\frac{1}{s}-\frac{1}{\frac{n}{2}-s}
+\sum_{k\in\Z^n\setminus\{0\}} \left(\pi
|k|^2\right)^{-s}\Gamma\left( s, \pi|k|^2\right)
\\&+\sum_{k\in\Z^n\setminus\{0\}} \left(\pi
|k|^2\right)^{-\frac{n}{2}+s}\Gamma\left( \frac{n}{2}-s,
\pi|k|^2\right),\nonumber
\end{align}where for $k\in\Z^n$,
$|k|^2=\sum_{i=1}^n k_i^2$. We need to obtain upper and lower bounds
for $\Gamma(\beta, x)$. Using integration by parts, it is easy to
show that when $\beta>0$,
\begin{align*}
\Gamma(\beta, x)=x^{\beta-1}e^{-x}\sum_{j=0}^m \frac{\prod_{i=1}^j
(\beta-i)}{x^j}+\left[\prod_{i=1}^{m+1}(\beta-i)\right]\int_x^{\infty}
t^{\beta-m-2}e^{-t}dt.
\end{align*}From this, we find that
\begin{align}\label{eq1015_8}
\Gamma(\beta,x)\leq x^{\beta-1}e^{-x}\sum_{j=0}^{[\beta]}
\frac{\prod_{i=1}^j (\beta-i)}{x^j}, \hspace{1cm}\Gamma(\beta,x)\geq
x^{\beta-1}e^{-x}\sum_{j=0}^{[\beta]+1} \frac{\prod_{i=1}^j
(\beta-i)}{x^j}.
\end{align}Therefore from    \eqref{eq1015_4}, we have
\begin{align*}
\Xi_9\left(\frac{9}{4}\right)\leq
-\frac{8}{9}+2\sum_{k\in\Z^9\setminus\{0\}}\left(\pi|k|^2\right)^{-1}e^{-\pi|k|^2}\sum_{j=0}^{2}
\frac{\prod_{i=1}^j \left(\frac{9}{4}-i\right)}{(\pi|k|^2)^j}.
\end{align*}For the summation over $k\in \Z^9\setminus\{0\}$, we
single out the 18 terms that contribute to $|k|=1$, i.e. we write
\begin{align}\label{eq1015_7}
&\sum_{k\in\Z^9\setminus\{0\}}\left(\pi|k|^2\right)^{-1}e^{-\pi|k|^2}\sum_{j=0}^{2}
\frac{\prod_{i=1}^j \left(\frac{9}{4}-i\right)}{(\pi|k|^2)^j}\\
=&\frac{18}{\pi}e^{-\pi}\sum_{j=0}^2\frac{\prod_{i=1}^j
\left(\frac{9}{4}-i\right)}{\pi^j}+\sum_{k\in\Z^9\setminus\{k\,:\,
|k|=0\;\text{or}\;
1\}}\left(\pi|k|^2\right)^{-1}e^{-\pi|k|^2}\sum_{j=0}^{2}
\frac{\prod_{i=1}^j
\left(\frac{9}{4}-i\right)}{(\pi|k|^2)^j}.\nonumber
\end{align}The first term of \eqref{eq1015_7} gives
\begin{align}\label{eq1015_6}
\frac{18}{\pi}e^{-\pi}\sum_{j=0}^2\frac{\prod_{i=1}^j
\left(\frac{9}{4}-i\right)}{\pi^j}=0.3540;
\end{align}whereas the second term of \eqref{eq1015_7} is
\begin{align}\label{eq1015_5}
&\sum_{k\in\Z^9\setminus\{k\,:\, |k|=0\;\text{or}\;
1\}}\left(\pi|k|^2\right)^{-1}e^{-\pi|k|^2}\sum_{j=0}^{2}
\frac{\prod_{i=1}^j \left(\frac{9}{4}-i\right)}{(\pi|k|^2)^j}\\\leq
& \frac{1}{2\pi}\sum_{j=0}^{2} \frac{\prod_{i=1}^j
\left(\frac{9}{4}-i\right)}{(2\pi
)^j}\sum_{k\in\Z^9\setminus\{k\,:\, |k|=0\;\text{or}\; 1\}}
e^{-\pi|k|^2}.\nonumber
\end{align}Using a simple inequality
\begin{align*}
\exp\left(-\pi \sum_{i=1}^n k_i^2\right)\leq
\exp\left(-\pi\sum_{i=1}^n |k_i|\right),
\end{align*}we find that
\begin{align*}
\sum_{k\in\Z^9\setminus\{k\,:\, |k|=0\;\text{or}\; 1\}}
e^{-\pi|k|^2}=&\sum_{k\in\Z^9} e^{-\pi|k|^2}-1-18e^{-\pi}\\
\leq & \sum_{k\in\Z^9} e^{-\pi\sum_{i=1}^9|k_i|}-1-18e^{-\pi}\\
=&\left(1+2\sum_{m=1}^{\infty}e^{-\pi m}\right)^9-1-18e^{-\pi}\\
=&\left(\frac{e^{\pi}+1}{e^{\pi}-1}\right)^9-1-18e^{-\pi}.
\end{align*}Therefore, the term \eqref{eq1015_5} is bounded above by
\begin{align*}
\frac{1}{2\pi}\sum_{j=0}^{2} \frac{\prod_{i=1}^j
\left(\frac{9}{4}-i\right)}{(2\pi
)^j}\left(\left(\frac{e^{\pi}+1}{e^{\pi}-1}\right)^9-1-18e^{-\pi}\right)=0.0769.
\end{align*}Combine with \eqref{eq1015_6}, we find that
\eqref{eq1015_7} is bounded above by
\begin{align*}
0.3540+0.0769=0.4309<\frac{4}{9}.
\end{align*}Therefore, $\Xi_9(9/4)<0$.
Now for $\Xi_{10}(5/2)$, a similar argument but using the lower
bound in \eqref{eq1015_8} for the incomplete gamma function gives
\begin{align*}
\Xi_{10}\left(\frac{5}{2}\right)>&-\frac{4}{5}+2\sum_{k\in\Z^{10}\setminus\{0\}}
\left(\pi|k|^2\right)^{-1}e^{-\pi|k|^2}\sum_{j=0}^{3}
\frac{\prod_{i=1}^j \left(\frac{5}{2}-i\right)}{(\pi|k|^2)^j}\\
>&-\frac{4}{5}+\frac{40}{\pi}e^{-\pi}\sum_{j=0}^{3}
\frac{\prod_{i=1}^j \left(\frac{5}{2}-i\right)}{\pi ^j}=0.04808>0.
\end{align*}

As a side remark, using  standard mathematics softwares such as
MATLAB and MATHEMATICA, it is easy to  compute $\Xi_n(s)$ to any
desired degree of accuracy from formula \eqref{eq1015_3} or
\eqref{eq1015_4}. Up to $10^{-12}$, we have
\begin{align*}
\Xi_9\left(\frac{9}{4}\right)=-0.065884758538,\\
\Xi_{10}\left(\frac{5}{2}\right) =0.205903040487.
\end{align*}

\end{proof}

Lemma \ref{l1}, Proposition \ref{p2} and the inequality
\eqref{eq1015_1} immediately implies that
\begin{proposition}\label{p5}
For all $n\geq 10$, the  sets
\begin{align}\label{eq1016_1}I_{n,+}=&\left\{ s\in(0, n/2)\,:\,
\Xi_{n}(s)>0\right\},\\ \label{eq1017_1} I_{n,-}=&\left\{ s\in(0,
n/2)\,:\, \Xi_{n}(s)<0\right\}\end{align}are open nonempty subsets
of $(0, n/2)$. Moreover, $n/4\in I_{n,+}$.
\end{proposition}In fact, from the functional equation
$\Xi_n(s)=\Xi_n((n/2)-s)$, one can even conclude that the intervals
$I_{n,+}$ and $I_{n,-}$ are symmetric with respect to the point
$n/4$, i.e.
$$\frac{n}{2}-I_{n,+}:=\left\{\frac{n}{2}-s\,:\, s\in I_{n,+}\right\}=I_{n,+}$$ and similarly for
$I_{n,-}$. Since $$\lim_{s\rightarrow 0^+}\Xi_n(s)=-\infty,$$ there
must exists an odd number of points $\gamma_{n,1}$, $\ldots$,
$\gamma_{n,2m_n+1}$ so that $0<\gamma_{n,1}<\gamma_{n,2}\leq
\gamma_{n,3}<\ldots<\gamma_{n, 2m_n}\leq \gamma_{n,2m_n+1}<n/4$ and
$$I_{n, +}= \bigcup_{i=1}^{m_n}\left( \gamma_{n,2i-1}, \gamma_{n,
2i}\right) \bigcup \left(\gamma_{n, 2m_n+1}, \frac{n}{2}-\gamma_{n,
2m_n+1}\right)\bigcup_{i=1}^{m_n}\left(\frac{n}{2}- \gamma_{n,2i},
\frac{n}{2}-\gamma_{n, 2i-1}\right).$$ For $n\leq 9$, Lemma \ref{l1}
is not sufficient to show that  $\Xi_n(s)<0$ for all $s\in (0,
n/2)$. Nevertheless, one can employ the same method as the proof in
Lemma \ref{l1} to show that
\begin{proposition}\label{p30}
$\Xi_9(s)<0$ for all $s\in (0, 9/2)$.
\end{proposition}

Then the inequality \eqref{eq1015_1} shows that
\begin{proposition}\label{p6}
For all $1\leq n\leq 9$, the function $\Xi_{n}(s)$ is negative for
all $s\in (0, n/2)$.
\end{proposition}

We give the proof of Proposition \ref{p30} here.
\begin{proof}By the reflection formula $\Xi_9(s)=\Xi_9((9/2)-s)$, it
is sufficient to show that $\Xi_9(s)<0$ for all $s\in (0, 9/4]$.
Let
\begin{align*}
\mathcal{Z}_1(s)=-\frac{1}{s}-\frac{1}{\frac{9}{2}-s}
\end{align*}and
\begin{align*}
\mathcal{Z}_2(s)=\int_1^{\infty}
t^{s-1}\left(\vartheta(t)^9-1\right)dt+\int_1^{\infty}
t^{\frac{9}{2}-s-1}\left(\vartheta(t)^9-1\right)dt,
\end{align*}so that $\Xi_9(s)=\mathcal{Z}_1(s)+\mathcal{Z}_2(s)$. By taking derivatives with respect to $s$,
it is easy to verify that as functions of $s$, $\mathcal{Z}_1(s)$ is
increasing on $(0,9/4]$ and $\mathcal{Z}_2(s)$ is decreasing on $[0,
9/4]$. Using the arguments in the proof of Lemma \ref{l1}, we find
that
\begin{align*}
\mathcal{Z}_2(0)\leq& \frac{18}{\pi}
e^{-\pi}+\frac{1}{2\pi}\left(\left(\frac{e^{\pi}+1}{e^{\pi}-1}\right)^9-1-18e^{-\pi}\right)
\\&+\frac{18}{\pi}e^{-\pi}\sum_{j=0}^{4} \frac{\prod_{i=1}^j
\left(\frac{9}{2}-i\right)}{\pi^j}+\frac{1}{2\pi}\sum_{j=0}^{4}
\frac{\prod_{i=1}^j \left(\frac{9}{2}-i\right)}{(2\pi
)^j}\left(\left(\frac{e^{\pi}+1}{e^{\pi}-1}\right)^9-1-18e^{-\pi}\right)\\=&1.2926.
\end{align*}Now we want to find an $x$ so that
$\mathcal{Z}_1(x)<-1.2926$. By try and error method, we find that
$x=0.95$ satisfies the required condition. In fact,
\begin{align*}
\mathcal{Z}_1(0.95)=-1.3343.
\end{align*}Therefore, for all $s\in (0, 0.95]$,
\begin{align*}
\Xi_9(s)=\mathcal{Z}_1(s)+\mathcal{Z}_2(s)\leq
\mathcal{Z}_1(0.95)+\mathcal{Z}_2(0)\leq -0.0417<0.
\end{align*}Repeating these steps, we have first
\begin{align*}
\mathcal{Z}_2(0.95)\leq & \frac{18}{\pi}
e^{-\pi}+\frac{1}{2\pi}\left(\left(\frac{e^{\pi}+1}{e^{\pi}-1}\right)^9-1-18e^{-\pi}\right)
\\&+\frac{18}{\pi}e^{-\pi}\sum_{j=0}^{3} \frac{\prod_{i=1}^j
\left(3.55-i\right)}{\pi^j}+\frac{1}{2\pi}\sum_{j=0}^{3}
\frac{\prod_{i=1}^j \left(3.55-i\right)}{(2\pi
)^j}\left(\left(\frac{e^{\pi}+1}{e^{\pi}-1}\right)^9-1-18e^{-\pi}\right)\\&=0.9728.
\end{align*}Then by try and error, we find that
\begin{align*}
\mathcal{Z}_1(1.55)=-0.9841.
\end{align*}Therefore, for all $s\in [0.95, 1.55]$,
\begin{align*}
\Xi_9(s)=\mathcal{Z}_1(s)+\mathcal{Z}_2(s)\leq
\mathcal{Z}_1(1.55)+\mathcal{Z}_2(0.95)\leq -0.0113<0.
\end{align*}Repeating again, we find that
\begin{align*}
\mathcal{Z}_2(1.55)\leq & \frac{18}{\pi}
e^{-\pi}\sum_{j=0}^{1}\frac{\prod_{i=1}^j
\left(1.55-i\right)}{\pi^j}+\frac{1}{2\pi}\sum_{j=0}^{1}\frac{\prod_{i=1}^j
\left(1.55-i\right)}{(2\pi)^j}\left(\left(\frac{e^{\pi}+1}{e^{\pi}-1}\right)^9-1-18e^{-\pi}\right)
\\&+\frac{18}{\pi}e^{-\pi}\sum_{j=0}^{2} \frac{\prod_{i=1}^j
\left(2.95-i\right)}{\pi^j}+\frac{1}{2\pi}\sum_{j=0}^{2}
\frac{\prod_{i=1}^j \left(2.95-i\right)}{(2\pi
)^j}\left(\left(\frac{e^{\pi}+1}{e^{\pi}-1}\right)^9-1-18e^{-\pi}\right)\\=&0.8943,
\end{align*}and\begin{align*}
\mathcal{Z}_1(2)=-0.9.
\end{align*}Therefore, for all $s\in [1.55, 2]$,
\begin{align*}
\Xi_9(s)=\mathcal{Z}_1(s)+\mathcal{Z}_2(s)\leq
\mathcal{Z}_1(2)+\mathcal{Z}_2(1.55)\leq -0.0057<0.
\end{align*}Finally,
\begin{align*}
\mathcal{Z}_2(2)\leq & \frac{18}{\pi}
e^{-\pi}\sum_{j=0}^{2}\frac{\prod_{i=1}^j
\left(2-i\right)}{\pi^j}+\frac{1}{2\pi}\sum_{j=0}^{2}\frac{\prod_{i=1}^j
\left(2-i\right)}{(2\pi)^j}\left(\left(\frac{e^{\pi}+1}{e^{\pi}-1}\right)^9-1-18e^{-\pi}\right)
\\&+\frac{18}{\pi}e^{-\pi}\sum_{j=0}^{2} \frac{\prod_{i=1}^j
\left(2.5-i\right)}{\pi^j}+\frac{1}{2\pi}\sum_{j=0}^{2}
\frac{\prod_{i=1}^j \left(2.5-i\right)}{(2\pi
)^j}\left(\left(\frac{e^{\pi}+1}{e^{\pi}-1}\right)^9-1-18e^{-\pi}\right)\\
=&0.8649,
\end{align*}and
\begin{align*}
\mathcal{Z}_1(2.25)=-0.8889,
\end{align*}which implies that for all $s\in[2, 2.25]$,
\begin{align*}\Xi_9(s)=\mathcal{Z}_1(s)+\mathcal{Z}_2(s)\leq
\mathcal{Z}_1(2.25)+\mathcal{Z}_2(2)\leq -0.0240<0.
\end{align*}This completes the proof of the proposition.
\end{proof}

\begin{figure}\centering \epsfxsize=.49\linewidth
\epsffile{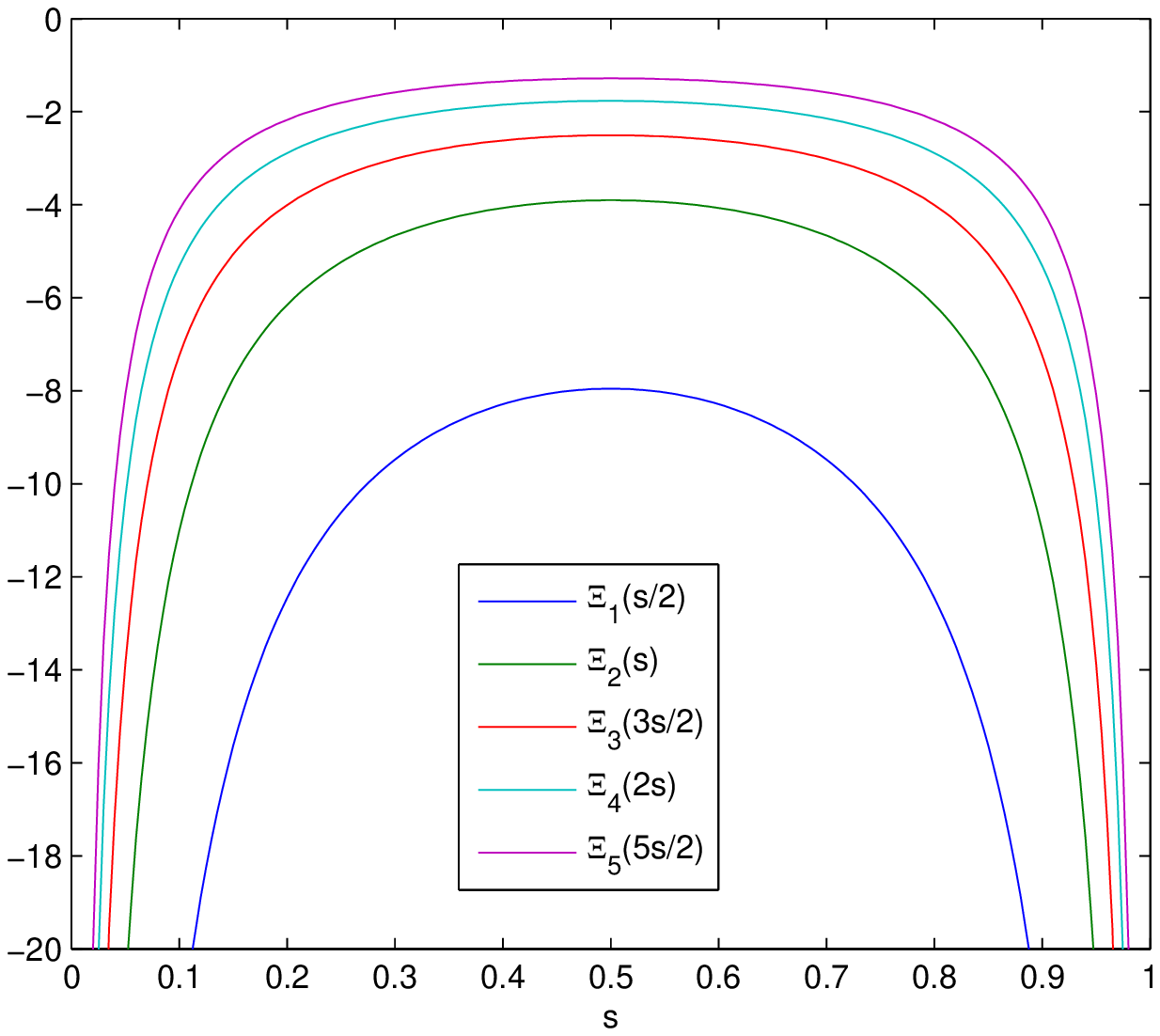}  \epsfxsize=.49\linewidth
\epsffile{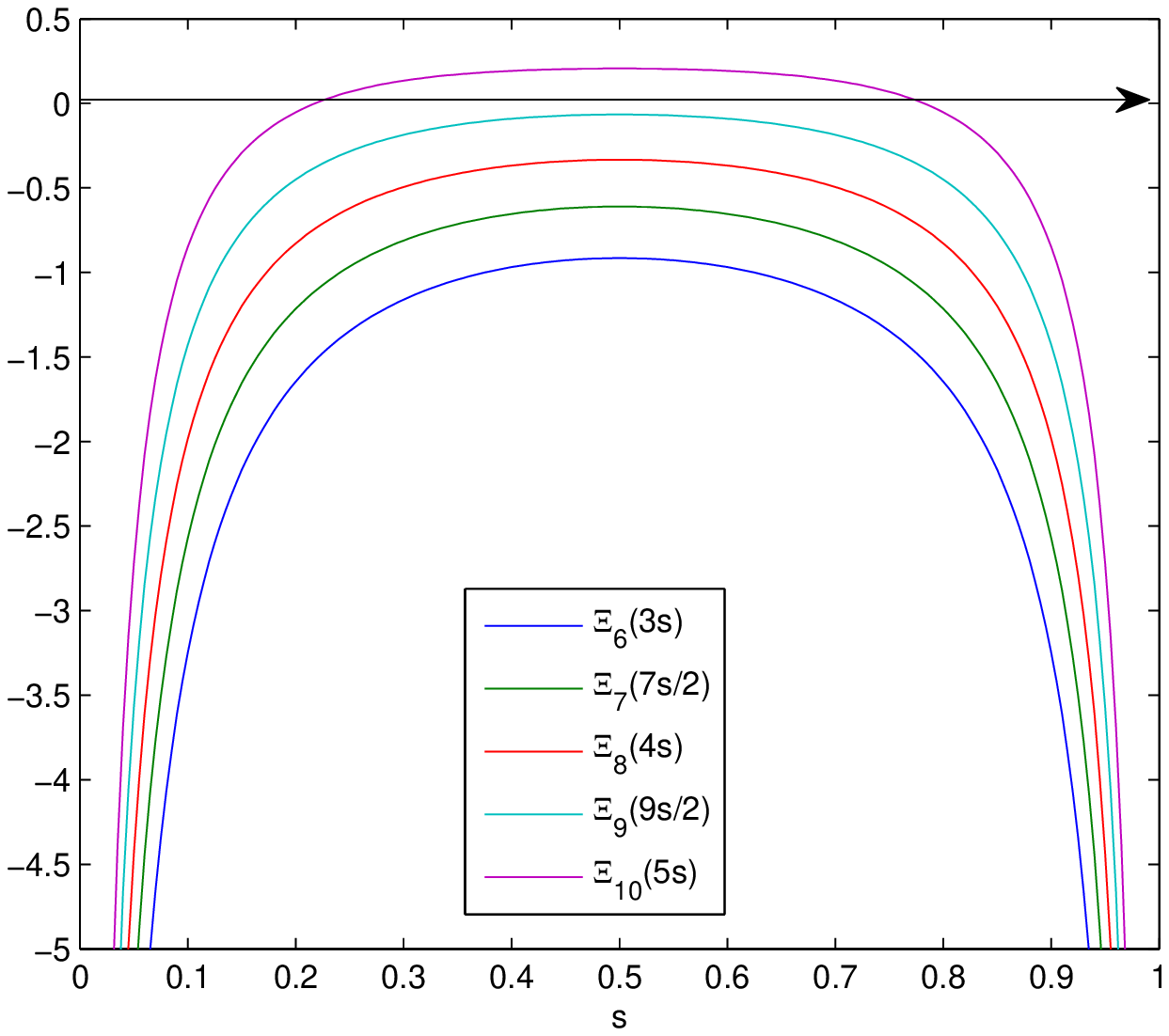}\caption{The graphs of
$\hat{\Xi}_n(s)=\Xi_n\left(\frac{ns}{2}\right)$ for $s\in (0,1)$.}
\end{figure}
 In Figure 1, we show the
graphs of $\hat{\Xi}_n(s)=\Xi_n(ns/2)$ for $s\in (0, 1)$ and $1\leq
n\leq 10$. From the graphs, we see that for $1\leq n\leq 10$,
$\hat{\Xi}_n(s)$ is increasing for $s\in (0, 1/2)$, decreasing for
$s\in (1/2, 1)$, and $s=1/2$ is a local maximum of $\hat{\Xi}_n(s)$.
One would tend to jump into the conclusion that this will be true
for all $n$. However, this is not the case. Since $\hat{\Xi}_n(s)$
satisfies the functional equation $\hat{\Xi}_n(s)=\hat{\Xi}_n(1-s)$,
it is easy to verify that $\hat{\Xi}_n'(1/2)=0$ and therefore
$s=1/2$ is indeed a local extremum of $\hat{\Xi}_n(s)$. In order to
determine the nature of this extremum, it is necessary to look at
the second derivative of $\hat{\Xi}_n(s)$. We obtain from
\eqref{eq1015_10} that
\begin{align}\label{eq1015_11}
\hat{\Xi}^{\prime\prime}_n(s)=&-\frac{4}{n}\left(\frac{1}{s^3}+\frac{1}{(1-s)^3}\right)
+\left(\frac{n}{2}\right)^2\int_1^{\infty} t^{\frac{ns}{2}-1}(\log
t)^2 \left(\vartheta(t)^n-1\right)dt\\&+\left(\frac{n}{2}\right)^2
\int_1^{\infty}t^{\frac{n}{2}(1-s)-1}\left(\vartheta(t)^n-1\right)dt.\nonumber
\end{align}From here, it is easy to verify that
\begin{align}\label{eq1015_12}\hat{\Xi}^{\prime\prime}_{n+1}(s)
>\hat{\Xi}^{\prime\prime}_n(s).\end{align} Using the same argument as
employed in showing that for any fixed $s\in (0, n/2)$,
$\hat{\Xi}_n(s)>0$ when $n$ is large enough, we can use
\eqref{eq1015_11} to prove that for any fixed $s\in (0,n/2)$,
$\Xi_n^{\prime\prime}(s)>0$ whenever $n$ is large enough. This
implies that the point $n/4$ will become a local minimum of
$\Xi_n(s)$ when $n$ is large enough. In fact, using numerical
computation, we find that
\begin{figure}\centering \epsfxsize=.49\linewidth
\epsffile{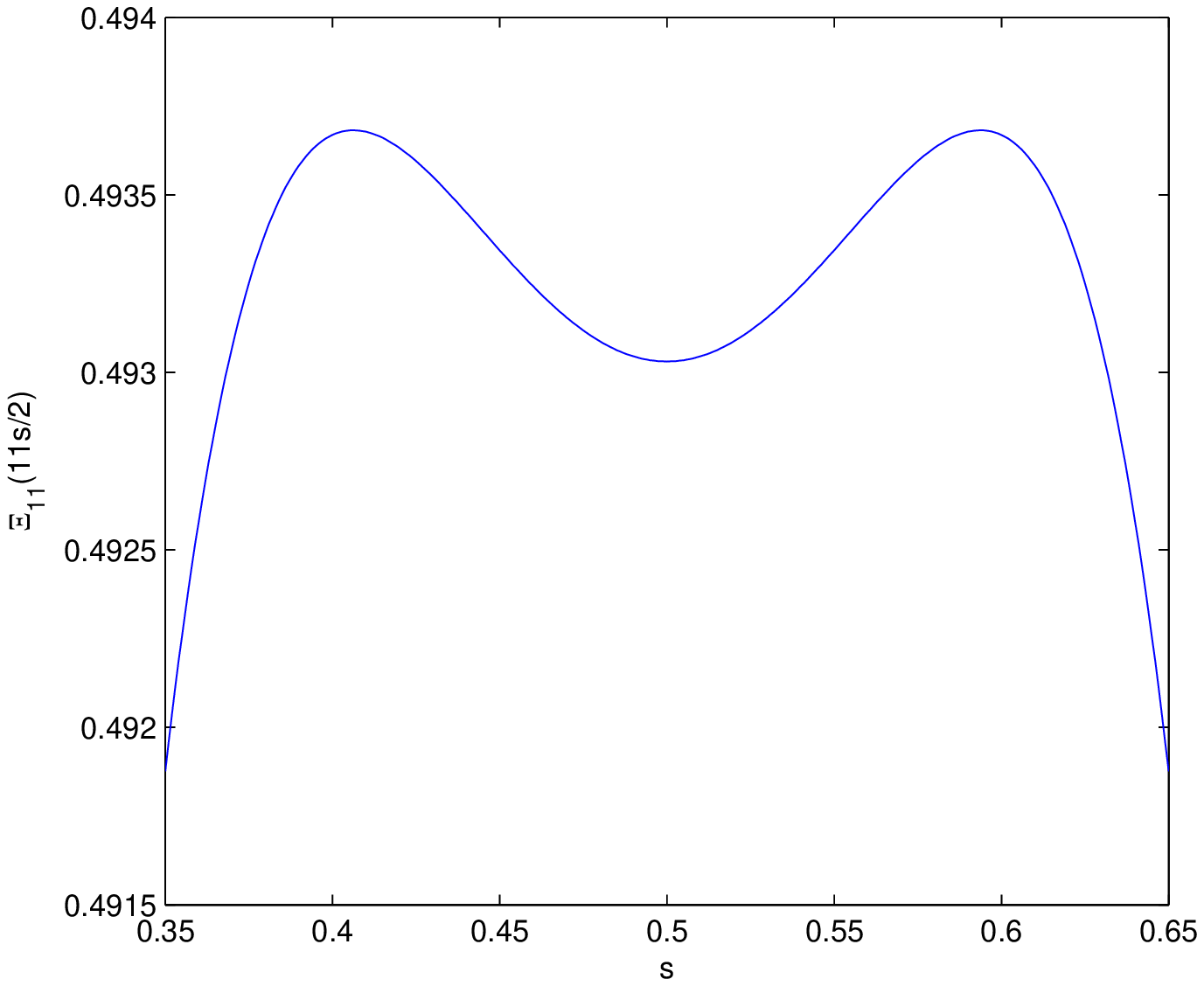}  \epsfxsize=.49\linewidth
\epsffile{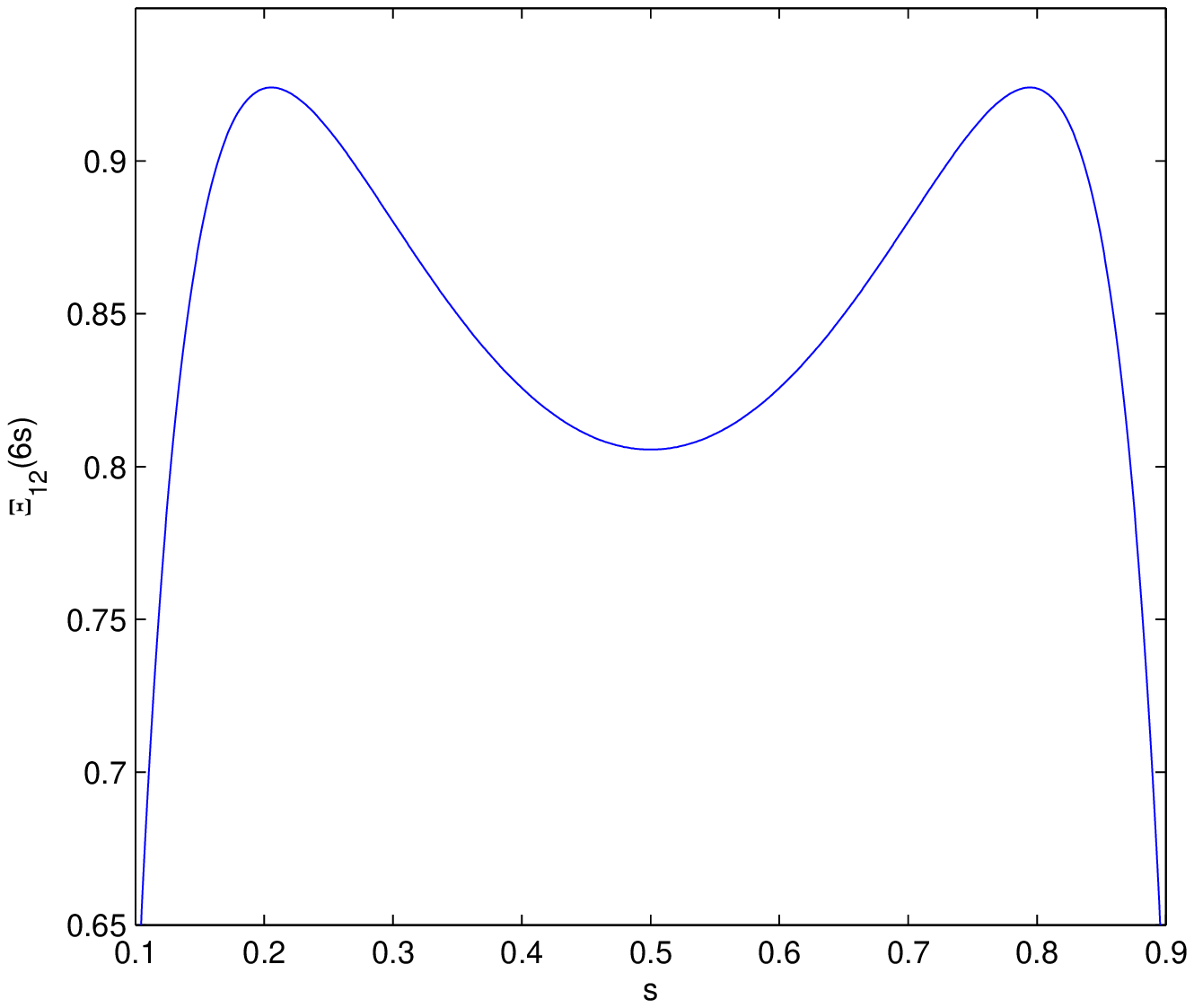} \caption{The graphs of
$\hat{\Xi}_{11}(s)=\Xi_{11}\left(\frac{11s}{2}\right)$ and
$\hat{\Xi}_{12}(s)=\Xi_{12}(6s)$.}
\end{figure}
\begin{align*}
\Xi^{\prime\prime}_{10}\left(\frac{5}{2}\right)=-0.101080515709,\\
\Xi^{\prime\prime}_{11}\left(\frac{11}{4}\right)=0.009568954836,
\end{align*}which together with \eqref{eq1015_12}, imply that for $1\leq n\leq 10$,
$s=n/4$ is a local maximum of $\Xi_n(s)$; whereas for $n\geq 11$,
$s=n/4$ is a local minimum of $\Xi_n(s)$. In Figure 2, we show
graphically that $s=1/2$ is a local minimum of $\hat{\Xi}_{11}(s)$
and $\hat{\Xi}_{12}(s)$.

 In Table 1, we tabulate the
interval $I_{n,+}, 10\leq n\leq 21$ where $\Xi_n(s)$ is positive.

\vspace{0.2cm} \noindent Table 1: The interval $I_{n,+}$ for $10\leq
n\leq 21$.

\vspace{0.2cm}\noindent
\begin{tabular}{|c|c||c|c|}
\hline
\hspace{0.5cm}$n$\hspace{0.5cm}  & $I_{n,+}$ & \hspace{0.5cm} $n$ \hspace{0.5cm} & $I_{n,+}$\\
\hline
 10 & \hspace{0.5cm}(1.0899, 3.9101)\hspace{0.5cm} & 11 & \hspace{0.5cm} (0.6401, 4.8599)\hspace{0.5cm} \\
 12 & (0.3976,    5.6024) & 13 & \hspace{0.5cm}(0.2498, 6.2502)\hspace{0.5cm} \\
 14 & (0.1562,    6.8438) & 15 & \hspace{0.5cm}(0.0964,    7.4036)\hspace{0.5cm} \\
 16 & (0.0585,    7.9415) & 17 & \hspace{0.5cm}(0.0348,   8.4652)\hspace{0.5cm} \\
 18 & (0.0202,    8.9798) & 19 & \hspace{0.5cm}(0.0115,    9.4885)\hspace{0.5cm} \\
 20 & (0.0064,    9.9936) & 21 & \hspace{0.5cm}(0.0034,   10.4966)\hspace{0.5cm} \\
 \hline
\end{tabular}

\vspace{0.5cm} One notices that for $10\leq n\leq 21$, $m_n=0$,
$I_{n,+}=(\gamma_{n,1}, (n/2)-\gamma_{n,1})$, $\gamma_{n,1}$ is
decreasing and $I_{n,+}\subseteq I_{n+1, +}$. In fact, from  the
formula \eqref{eq1015_3}, it is easy to verify that for $s\in (0,
n/2)$,
$$\Xi_{n+1}(s)>\Xi_n(s).$$ This shows that $I_{n,+}\subseteq I_{n+1,
+}$ for all $n\geq 10$. From this table, it is also natural to
conjecture that $I_{n,+}$ is an open connected interval with center
at $n/4$ for all $n\geq 10$. However, to prove this would require a
very detailed analysis of $\Xi_n(s)$ and its higher derivatives. We
do not intend to deal further into this problem here.

We now return to the discussion about the sign of the general
Xi-function $\Xi_n(s;a_1,\ldots, a_n)$ in the range $s\in (0, n/2)$.
Since we have shown in Theorem \ref{th1} that for any fixed $s$, the
point $a_1=\ldots=a_n$ is the minimum  of $\Xi_n(s;a_1,\ldots, a_n),
(a_1,\ldots, a_n)\in(\R^n)^+$ with $\prod_{i=1}^n a_n$ fixed, now
together with Proposition \ref{p5}, Proposition \ref{p6} and
Proposition \ref{p2}, we obtain
\begin{proposition}
For $1\leq n\leq 9$, there exists a nonempty region containing the
ray $a_1=\ldots=a_n$ in $(\R^n)^+$ where $\Xi_n(s;a_1,\ldots, a_n)$
is negative for all $s\in (0,n/2)$.
\end{proposition}

\begin{proposition}
If $n\geq 10$, then for every $(a_1,\ldots, a_n)\in (\R^+)^n$, there
exists a nonempty open subset of $(0, n/2)$ where $\Xi_n(s;
a_1,\ldots, a_n)$ is positive and a nonempty open subset of $(0,
n/2)$ where $\Xi_n(s; a_1,\ldots, a_n)$ is negative.

\end{proposition}This proposition implies that
generalized Riemann hypothesis is not true for all the Epstein zeta
function of the form $Z_{n}(s; a_1,\ldots, a_n)$ (i.e. Epstein zeta
function of positive definite diagonal quadratic forms) when $n\geq
10$.

Now consider the Epstein zeta function as a function of
$(a_1,\ldots, a_n)\in(\R^+)^n$ with $s$ fixed, we can conclude with
the help of Propositions \ref{p4} and \ref{p6} that
\begin{proposition}\label{p7}
If $1\leq n\leq 9$,  then for any fixed $s\in (0, n/2)$, there
exists a nonempty open region $\Omega_{s,n}^+$ of $(a_1,\ldots,
a_n)\in(\R^+)^n$ where $\Xi_{n}(s; a_1,\ldots, a_n)$ is positive and
a nonempty open region $\Omega_{s,n}^-$ of $(a_1,\ldots,
a_n)\in(\R^+)^n$ where $\Xi_{n}(s; a_1,\ldots, a_n)$ is negative.
The ray $a_1=\dots=a_n$ lies inside the region $\Omega_{s,n}^-$.
\end{proposition}

\begin{proposition}\label{p8}
If $ n\geq 10$,  then if  $s\in I_{n,+}$ (eq. \eqref{eq1016_1}), the
function $\Xi_n(s; a_1,\ldots, a_n)$ is positive for all $
(a_1,\ldots, a_n)\in (\R^+)^n$; whereas if $s\in I_{n,-}$ (eq.
\eqref{eq1017_1}), there exists a nonempty open region
$\Omega_{s,n}^+$ of $(a_1,\ldots, a_n)\in(\R^+)^n$ where $\Xi_{n}(s;
a_1,\ldots, a_n)$ is positive and a nonempty open region
$\Omega_{s,n}^-$ of $(a_1,\ldots, a_n)\in(\R^+)^n$ where $\Xi_{n}(s;
a_1,\ldots, a_n)$ is negative. In the latter case, the ray
$a_1=\dots=a_n$ lies inside the region $\Omega_{s,n}^-$.
\end{proposition}

  Propositions \ref{p7} and \ref{p8} are important for
our study of symmetry breaking of interacting fractional
Klein--Gordon field in \cite{e10}. In the following section, it will
be shown  that the region $\Omega_{s,n}^-$ is connected. Moreover,
under the map $(a_1,\ldots, a_n)\in (\R^+)^n\mapsto (\log
a_1,\ldots, \log a_n)\in \R^n$, the projection of its image to the
plane $x_1+\ldots+x_n=$ constant is a convex set.
\section{Convexity of the Epstein zeta function}

In this section, we   show that \begin{theorem}\label{th2}For any
fixed $s$, when $\sum_{i=1}^n c_i =c$ is fixed,  the function
\begin{align}\label{eq1017_2}(c_1,\ldots, c_n) \mapsto &\Xi_n\left(s; e^{c_1},\ldots,
e^{c_n}\right)\nonumber\\&=-\frac{e^{\frac{c}{2}}}{s}-\frac{e^{-\frac{c}{2}}}{\frac{n}{2}-s}
+e^{\frac{c}{2}}\int_1^{\infty}t^{s-1}\left(\prod_{i=1}^n
\vartheta(e^{2c_i}t)-1\right)dt\\
&+e^{-\frac{c}{2}}\int_1^{\infty}t^{\frac{n}{2}-s-1}\left(\prod_{i=1}^n
\vartheta(e^{-2c_i}t)-1\right)dt,\nonumber\end{align}regarded as a
function of any $(n-1)$ of the variables  $\{c_1,\ldots,c_n)$, is a
convex function.\end{theorem} Recall that a function $f(x)$, $x\in
\R^n$ is convex if and only if for any $\lambda\in (0,1)$ and any
two points $x, y\in\R^n$,
\begin{align*}
f\left(\lambda x+(1-\lambda)y\right)\leq \lambda
f(x)+(1-\lambda)f(y).
\end{align*}If $f$ has continuous second partial derivatives, then
$f$ is convex if and only if the Hessian of $f$,
\begin{align*}
H_f(x_1,\ldots, x_n)=\left(\frac{\pa^2f}{\pa x_i\pa
x_j}\right)_{\substack{1\leq i\leq n\\1\leq j\leq n}},
\end{align*} is a positive semi-definite matrix. One way to prove that a Hermitian matrix is positive semi-definite is
to use the Sylvester criterion, which says that an $n\times n$
Hermitian matrix is positive semi-definite if and only if all its
leading principal minors are nonnegative. For $1\leq i\leq n$, the
$i^{\text{th}}$ leading principal minor of an $n\times n$ matrix $A$
is the determinant of the $i\times i$ square matrix on the upper
left corner of $A$.

 In the
formula \eqref{eq1017_2},  the  dependence   on $(c_1,\ldots, c_n)$
only comes from the terms $\prod_{i=1}^n\vartheta(e^{2c_i}t)$ and
$\prod_{i=1}^n\vartheta(e^{-2c_i}t)$. Notice that if a function
$f(x_1,\ldots, x_n)$ is convex, then under any affine change of
coordinates $$x_i=\sum_{j=1}^n l_{ij} y_j+d_i, \hspace{1cm} 1\leq
i\leq n,$$ the function $$\hat f(y_1,\ldots,
y_n)=f\left(\sum_{j=1}^n l_{1j}y_j+d_1,\ldots,\sum_{j=1}^n
l_{nj}y_j+d_n\right)$$ is also convex. Moreover, the projection of
$f$ to any hyperplane $\sum_{i=1}^n b_i x_i=b$ is also convex.
Therefore, if the function
\begin{align*}
\Theta_n(x_1,\ldots, x_n)=\prod_{i=1}^n
\vartheta\left(e^{x_i}\right)=\prod_{i=1}^n \Theta_1(x_i),
\hspace{1cm}(x_1,\ldots, x_n)\in\R^n
\end{align*}is convex, then $\prod_{i=1}^n\vartheta(e^{2c_i}t)$ and
$\prod_{i=1}^n\vartheta(e^{-2c_i}t)$, regarded as  functions of
$\{c_1,\ldots,c_{n-1}\}$ (so that $c_n =c-\sum_{i=1}^{n-1}c_i$),
 are convex. This will show that
$\Xi_n(s; e^{c_1},\ldots, e^{c_n})$ is a convex function of
$(c_1,\ldots,c_{n-1})$. Therefore to prove Theorem \ref{th2}, we
only need to show that for all $n\geq 1$, the function
$\Theta_n(x_1,\ldots, x_n)$ is convex.
 For $n= 1$, we use the following fact:
\begin{lemma}\label{l10}
If $g:\R\rightarrow \R$ is a  convex function, then the function $f$
defined by $$f(x)=e^{g(x)},$$is also convex.
\end{lemma}\begin{proof}This is a simple consequence of the fact that
the exponential function $e^x$ is a convex function, and composition
of convex functions are convex.
\end{proof}

Using this lemma, we can immediately conclude from Proposition
\ref{p1} that
\begin{lemma}\label{l3}
The function $$\Theta_1(x)=\vartheta\left(e^x\right),
\hspace{1cm}x\in\R$$ is convex.
\end{lemma}

For $n\geq 2$, in view of the definition of $\Theta_n$, we are led
to the study of a general function of the type
\begin{align}\label{eq1017_11}
\mathcal{F}_n(x_1, \ldots, x_n)=\prod_{i=1}^n f(x_i),
\end{align}where $f:\R\rightarrow \R^+$ is a positive function of one variable.
For such functions, its Hessian is given by
\begin{align*}
H_{\mathcal{F}_n}(x_1,\ldots,x_n) =\mathcal{F}_n(x_1,\ldots,
x_n)^n\begin{pmatrix} \frac{f^{\prime\prime}(x_1)}{f(x_1)} &
\frac{f'(x_1)}{f(x_1)}\frac{f'(x_2)}{f(x_2)}&\ldots &
\frac{f'(x_1)}{f(x_1)}\frac{f'(x_n)}{f(x_n)}\\
\frac{f'(x_2)}{f(x_2)}\frac{f'(x_1)}{f(x_1)} &
\frac{f^{\prime\prime}(x_2)}{f(x_2)}&\ldots &
\frac{f'(x_2)}{f(x_2)}\frac{f'(x_n)}{f(x_n)}\\
\vdots &\vdots & & \vdots\\
\frac{f'(x_n)}{f(x_n)}\frac{f'(x_1)}{f(x_1)}&\frac{f'(x_n)}{f(x_n)}\frac{f'(x_2)}{f(x_2)}
&\ldots &\frac{f^{\prime\prime}(x_n)}{f(x_n)}\end{pmatrix}.
\end{align*}By the Sylvester criterion, to show that $\mathcal{F}_n$
is convex, it is sufficient to show that all the leading principal
minors of
\begin{align}\label{eq1017_12}
H_n(x_1,\ldots, x_n)=\begin{pmatrix}
\frac{f^{\prime\prime}(x_1)}{f(x_1)} &
\frac{f'(x_1)}{f(x_1)}\frac{f'(x_2)}{f(x_2)}&\ldots &
\frac{f'(x_1)}{f(x_1)}\frac{f'(x_n)}{f(x_n)}\\
\frac{f'(x_2)}{f(x_2)}\frac{f'(x_1)}{f(x_1)} &
\frac{f^{\prime\prime}(x_2)}{f(x_2)}&\ldots &
\frac{f'(x_2)}{f(x_2)}\frac{f'(x_n)}{f(x_n)}\\
\vdots &\vdots & & \vdots\\
\frac{f'(x_n)}{f(x_n)}\frac{f'(x_1)}{f(x_1)}&\frac{f'(x_n)}{f(x_n)}\frac{f'(x_2)}{f(x_2)}
&\ldots &\frac{f^{\prime\prime}(x_n)}{f(x_n)}\end{pmatrix}
\end{align}are nonnegative. Notice that the $i^{\text{th}}$ leading
principal minor of $H_n(x_1,\ldots,x_n)$ is the determinant of
$H_i(x_1,\ldots, x_i)$. Therefore, one only have to show that $\det
H_n\geq 0$ for all $n\geq 1$. For $f(x)=\vartheta(e^x)$, the $n=1$
case is already proved in Lemma \ref{l3}. For $n\geq 2$, to further
simplify our problem, notice that for fixed $(x_1,\ldots,x_n)\in
\R^n$, the matrix $H_n$ has the form
\begin{align*}
J_n =\begin{pmatrix} z_1& w_1 w_2 & \ldots & w_1 w_n\\
w_2w_1 & z_2 &\ldots & w_2w_n\\
\vdots &\vdots & & \vdots\\
w_n w_1 & w_n w_2 & \ldots & z_n\end{pmatrix}
\end{align*}for some $2n$ real numbers $z_1,\ldots, z_n$ and
$w_1,\ldots, w_n$. When $n=2$, the determinant of $J_2$ is
\begin{align}\label{eq1017_6}
\det J_2= z_1 z_2- w_1^2
w_2^2=(z_1-w_1^2)(z_2-w_2^2)+w_1^2(z_1-w_1^2)+w_2^2(z_2-w_2^2).
\end{align}The reason for expressing the determinant in this form is
inspired by the fact that for
$z_i=\frac{f^{\prime\prime}(x_i)}{f(x_i)}$,
$w_i=\frac{f'(x_i)}{f(x_i)}$,
\begin{align}\label{eq1017_13}z_i-w_i^2=\frac{d^2}{dx_i^2}\log f(x_i).\end{align}Eq.
\eqref{eq1017_6} implies that if $z_i-w_i^2\geq 0$, then $\det
J_2\geq 0$. This gives
\begin{lemma}\label{l4}
Let $f:\R\rightarrow \R^+$ be a twice continuously differentiable
positive function. If  $\log f(x)$ is a convex function, then the
function $\mathcal{F}_2(x_1, x_2)=f(x_1)f(x_2)$ is convex.
\end{lemma}
\begin{proof}
Recall that $\mathcal{F}_2(x_1, x_2)$ is convex if and only if
$$\det H_1(x_1)=\frac{f^{\prime\prime}(x_1)}{f(x_1)}\geq 0$$ and
$$\det H_2(x_1, x_2)=\det
J_2\left[z_1=\frac{f^{\prime\prime}(x_1)}{f(x_1)},
z_2=\frac{f^{\prime\prime}(x_2)}{f(x_2)}; w_1=\frac{f^{\prime
}(x_1)}{f(x_1)}, w_2=\frac{f^{\prime }(x_2)}{f(x_2)}\right]\geq 0.$$
As discussed above, if $\log f(x)$ is convex, then $\det H_2(x_1,
x_2)\geq 0$. On the other hand, Lemma \ref{l10} implies that $f(x)$
is also convex. Therefore, $\det H_1(x_1)\geq 0$. This completes the
proof.
\end{proof}

For the case under consideration, $f(x)=\vartheta(e^x)$, and it is
shown in Proposition
 \ref{p1} that  $\log \vartheta(e^x)$ is a
 convex function. Therefore, Lemma \ref{l4} implies that
 \begin{lemma}\label{l5}
The function $$\Theta_2(x_1,
x_2)=\vartheta(e^{x_1})\vartheta(e^{x_2}), \hspace{1cm}(x_1, x_2)\in
\R^2$$ is convex.
 \end{lemma}

Returning to the general case. First we claim the following, which
is a generalization of \eqref{eq1017_6}.
\begin{proposition}\label{p10}
For $n\geq 2$,
\begin{align}\label{eq1017_14}
\det J_n =\det \begin{pmatrix} z_1& w_1 w_2 & \ldots & w_1 w_n\\
w_2w_1 & z_2 &\ldots & w_2w_n\\
\vdots &\vdots & & \vdots\\
w_n w_1 & w_n w_2 & \ldots & z_n\end{pmatrix}=\prod_{i=1}^n
(z_i-w_i^2)+\sum_{i=1}^n w_i^2\left[\prod_{j\neq
i}\left(z_j-w_j^2\right)\right].
\end{align}Therefore, if $z_i\geq w_i^2$ for all $1\leq i\leq n$, then $\det
J_n\geq 0$.
\end{proposition}We would like to discuss the consequence of this
proposition first, and give the proof later. From the above
discussion, the function $\mathcal{F}_n$ of the form
\eqref{eq1017_11} is convex if and only if all the determinants of
$H_i$ (eq. \eqref{eq1017_12}), $1\leq i\leq n$ are nonnegative. Now
$H_n$ is a matrix of the form $J_n$, with $z_i -w_i^2$ given by
\eqref{eq1017_13}. Therefore, we conclude from Proposition \ref{p10}
that
\begin{proposition}
Let $f:\R\rightarrow \R^+$ be a twice continuously differentiable
positive function. If  $\log f(x)$ is a  convex function, then the
function $$\mathcal{F}_n(x_1,\ldots, x_n)=\prod_{i=1}^nf(x_i)$$ is
convex.
\end{proposition}As in Lemma \ref{l5},   this
proposition implies that
\begin{proposition}
For any $n\geq 1$, the function $$\Theta_n(x_1,\ldots,
x_n)=\prod_{i=1}^n \vartheta(e^{x_i}), \hspace{1cm} (x_1,\ldots,
x_n)\in \R^n,$$ is a convex function.
\end{proposition}As discussed in the beginning of this section, this
concludes the proof of Theorem \ref{th2}. Therefore, what is left is
the proof of Proposition \ref{p10}.

\begin{proof}By factoring out $w_i$ from row $i$ and column $i$, we have
\begin{align*}
\det J_n =\left[\prod_{i=1}^n w_i^2\right]^2
\mathcal{J}_n\left(\frac{z_1}{w_1^2},\frac{z_2}{w_2^2},\ldots,
\frac{z_n}{w_n^2}\right),
\end{align*}where \begin{align*}
\mathcal{J}_n(\alpha_1,\ldots, \alpha_n)=\det \begin{pmatrix}\alpha_1 & 1  &1 &\ldots & 1\\
1 &\alpha_2 & 1 &\ldots & 1\\
1 & 1 & \alpha_3 & \ldots & 1\\
\vdots & \vdots & \vdots & &\vdots\\
1 & 1 & 1 & \ldots & \alpha_n\end{pmatrix}.
\end{align*}By elementary row operation, we find that
\begin{align*}
\mathcal{J}_n(\alpha_1,\ldots, \alpha_n)=&\det \begin{pmatrix}\alpha_1 & 1  &1 &\ldots & 1\\
1 &\alpha_2 & 1 &\ldots & 1\\
1 & 1 & \alpha_3 & \ldots & 1\\
\vdots & \vdots & \vdots & &\vdots\\
1 & 1 & 1 & \ldots & \alpha_n\end{pmatrix}\\=&\det \begin{pmatrix}\alpha_1-1 & 1-\alpha_2  &0 &\ldots & 0\\
0 &\alpha_2-1 & 1-\alpha_3 &\ldots & 0\\
0 & 0 & \alpha_3-1 & \ldots & 0\\
\vdots & \vdots & \vdots & &\vdots\\
1 & 1 & 1 & \ldots & \alpha_n\end{pmatrix}\\
=&\left(\alpha_1-1\right)\mathcal{J}_{n-1}\left(\alpha_2,\ldots,
\alpha_n\right)+(-1)^{n-1}\prod_{i=2}^n(1-\alpha_i).
\end{align*}It then follows by induction that
\begin{align*}
\mathcal{J}_n(\alpha_1,\ldots, \alpha_n)=\prod_{i=1}^n (\alpha_i-1)
+\sum_{i=1}^n \prod_{j\neq i}(\alpha_j-1).
\end{align*}This proves the proposition.
\end{proof}

Next we discuss the consequences of Theorem \ref{th2}. Let $j$ be a
positive integer less than or equal to $n-1$, $A$  an $n\times j$
matrix such that the sum of every column vanishes, i.e.,
\begin{align}\label{eq1018_1}\sum_{i=1}^n A_{ij}=0, \end{align}and let $v$ be a vector in $\R^n$.
Define a $j$-variable function ${\Xi}_{n,A,v}(s; b_1,\ldots, b_j)$
by
\begin{align*}
{\Xi}_{n,A,v}(s; b_1,\ldots, b_j) = \Xi_n\left(s;
\exp\left[v_1+\sum_{l=1}^j A_{1l}b_l \right],\ldots,
\exp\left[v_n+\sum_{l=1}^j A_{nl}b_l \right]\right).
\end{align*}Notice that the condition \eqref{eq1018_1} implies
\begin{align*}
\prod_{i=1}^n \exp\left[v_i+\sum_{l=1}^j A_{il}b_l
\right]=\exp\left[\sum_{i=1}^n v_i\right]=\text{constant}.
\end{align*}Since affine change of coordinates does not affect the
convexity of a function, we conclude from Theorem \ref{th2} that
$\Xi_{n,A,v}(s;b_1,\ldots, b_j)$ is a convex function of
$(b_1,\ldots, b_j)$. As an example, if $v=0$ and $A$ is the $n\times
(n-1)$ matrix
\begin{align}\label{eq1018_3}
A=\begin{pmatrix} 1 & 0   &\ldots & 0\\
0 & 1  & \ldots & 0\\
\vdots & \vdots   & &\vdots\\
0 & 0   & \ldots & 1\\
-1 & -1   & \ldots & -1\end{pmatrix},
\end{align}then $\Xi_{n,A,v}(b_1,\ldots, b_n)$ is just the function defined in
\eqref{eq1017_2} with $b_i=c_i$ for $1\leq i\leq n-1$ and
$\sum_{i=1}^n c_n=0$.

By the convexity of $\Xi_{n,A,v}\left(s; b_1,\ldots, b_j\right)$, we
can conclude that the region $\Omega_{s, A, v}^-$ of $(b_1,\ldots,
b_j)\in \R^j$ where $\Xi_{n,A,v}\left(s; b_1,\ldots, b_j\right)<0$
is a convex, and therefore connected region. Using the fact that at
fixed $\prod_{i=1}^na_i$, the minimum of $\Xi_n(s; a_1,\ldots, a_n)$
appears at $a_1=\ldots=a_n$, one can even conclude that if
$\Omega_{s, A, v}^-$ is nonempty and if $b=\hat{b}$ is a solution of
the system
\begin{align}\label{eq1018_4}
Ab+v=\left[\frac{1}{n}\sum_{i=1}^n
v_i\right]\begin{pmatrix}1\\1\\\vdots\\1\end{pmatrix},
\end{align} then $\hat{b}\in
\Omega_{s,A,v}^-$. When $A$ is given by \eqref{eq1018_3} and $v=0$,
the system \eqref{eq1018_4} has a unique solution $b=0$. Therefore,
the region of $(c_1,\ldots, c_{n-1})\in \R^n$ where the function
$\Xi_n\left(s; e^{c_1}, \ldots, e^{c_{n-1}},
e^{c-c_1-\ldots-c_{n-1}}\right)$ \eqref{eq1017_2} is negative is a
convex connected region containing the origin
$c_1=\ldots=c_{n-1}=0$.

As a second example, suppose  the variables $(a_1,\ldots, a_n)$ in
$\Xi_n(s; a_1,\ldots, a_n)$ are such that $a_1:\ldots : a_n
=1:k_2:\ldots:k_n$ and $\prod_{i=1}^n a_i=1$. A simple computation
shows that under these conditions, $a_1,\ldots, a_n$ can be
expressed as functions of $(k_2, \ldots, k_n)$:
\begin{align*}
a_1=\frac{1}{ \prod_{i=2}^n k_i^{\frac{1}{n}}},
\hspace{0.5cm}a_2=\frac{k_2^{\frac{n-1}{n}}}{\prod_{i=3}^n
k_i^{\frac{1}{n}}},\hspace{0.5cm}\ldots,\hspace{0.5cm}
a_n=\frac{k_n^{\frac{n-1}{n}}}{\prod_{i=2}^{n-1}k_i^{\frac{1}{n}}}.
\end{align*}Therefore the function $\Xi_n(s; a_1,\ldots, a_n)$ can
be regarded as a function of $\log k_i$, $2\leq i\leq n$ with
corresponding $A$ and $v$ given by
$$A=\begin{pmatrix}-\frac{1}{n} & -\frac{1}{n} & \ldots & -\frac{1}{n}\\
\frac{n-1}{n} & -\frac{1}{n} & \ldots & -\frac{1}{n}\\
\vdots & \vdots & & \vdots\\
-\frac{1}{n} & -\frac{1}{n} & \ldots &\frac{n-1}{n}
\end{pmatrix}$$ and $v=0$. Consequently, the region $\Omega_{s,n}^-$
where $\Xi_n(s; a_1,\ldots, a_n)<0$, plotted with respect to the
variables $\log k_2,\ldots, \log k_n$ is a convex and connected
region containing the point $\log k_2=\ldots=\log k_n=0$. Reducing
the number of variables by setting some of the variables $k_i$ to be
equal to a constant or setting $k_i=k_j$ for some pairs of $i\neq j$
are tantamount to restricting the variables $(\log k_2, \ldots, \log
k_n)$ to the intersections of hyperplanes in $\R^{n-1}$. Therefore
plotting the region $\Omega_{s,n}^-$ with respect to the   remaining
$\log k_i$ variables, the new region is still convex and connected.
In \cite{e10}, we have plotted the regions where $\Xi_3(s; a_1, a_2,
a_3)<0$ with respect to the variables $\log k_2$ and $\log k_3$, and
the regions where $\Xi_4(s; a_1,a_2,a_3, a_4)<0$ with respect to the
variables $\log k_2, \log k_3$, with $k_4=1$ and $k_4=3$, for some
values of $s$. The graphs show that these regions are indeed convex
and connected. As a matter of fact, when the regions
$\Omega_{s,n}^-$ are plotted using computer softwares, we can only
determine the regions  in finite domains of the $\log k$ variables.
Our convexity and connectivity results guarantee that there does not
exists region of $\Omega_{s,n}^-$ outside the finite domain we
consider.

Finally, notice that the map $(a_1,\ldots, a_n)\in (\R^+)^n \mapsto
(\log a_1,\ldots, \log a_n)\in \R^n$ is continuous. On the other
hand, we have shown that the intersection of the region
$\Omega_{s,n}^-$ where $\Xi_{n}(s; a_1,\ldots, a_n)<0$ with any
hypersurface in $(\R^+)^n$ of the form $\prod_{i=1}^n a_i=a$ a
constant, contains the point $a_1=\ldots=a_n=a^{\frac{1}{n}}$. These
allow one to conclude that the region $\Omega^-_{s,n}$, as a region
of $(a_1,\ldots, a_n)\in (\R^+)^n$, is a connected region that
contains the ray $a_1=\ldots=a_n$.

\section{Concluding remarks}

This paper is motivated by  our recent work \cite{e10}, in which we
need to determine the conditions for which symmetry breaking occurs
in the $\lambda\varphi^4$--interacting fractional Klein--Gordon
field theory on a toroidal spacetime $T^n\times \R^{N}$. The results
obtained in this work are used in \cite{e10}. Since the Epstein zeta
function $Z_n(s; a_1,\ldots, a_n)$ always appears when one studies
field theories on toroidal manifolds or rectangular cavities, the
results presented here may have  potential applications for other
works along these directions. For physics applications, one is
usually interested in simple toroidal manifolds of the form
$T^n\times \R^N$, which can be considered as  quotients of $\R^n$ by
rectangular lattices. This explains why we consider Epstein zeta
function of the form $Z_n(s; a_1,\ldots, a_n), a\in (\R^+)^n$
instead of the general Epstein zeta function $Z_n(A;s)$ with $A$ any
positive definitive symmetric matrix. An advantage of this
simplification is that we can determine the minimum of $Z_{n}(s;
a_1,\ldots, a_n), a\in (\R^+)^n, \prod_{i=1}^n a_i=1$ for all
$s\in\R\setminus \{0,n/2\}$ and all $n\in \mathbb{N}$.  In contrast,
for general Epstein zeta function $Z_n(A;s)$ with $\det A=1$, only
some local minima have been determined for $n=2,3,4,5,6,7,8,24$ with
$s$ in some range of $\R$ [15--38]. In fact, for all these known
cases, the minimum of $Z_n(s; a_1, \ldots, a_n), a\in (\R^+)^n,
\prod_{i=1}^n a_i=1$ is no longer a local minimum in the larger
domain where $A\in \left\{\text{positive definite
symmetric}\;n\times n\;\text{matrices}\right\}$. In the attempt to
search for the correct spacetime model, it might turn out that the
general toroidal manifold $\R^n/L$, where $L$ is any lattice in
$\R^n$, will be of importance to physics. In that case, we need to
extend the work of this paper to general Epstein zeta function
$Z_n(A;s)$. It is hoped that the methods and results in this paper
will be useful for the study of this general problem.

Our result about the convexity of the Epstein zeta function $Z_n(s;
a_1,\ldots, a_n)$, as a function of $\log a_1,\ldots, \log a_n$ with
$\log a_1+\ldots+\log a_n$ fixed,   seems to be new. This result is
important for determining the regions where $Z_n(s; a_1,\ldots,
a_n)$ is positive or negative, as discussed at the end of section 4.
As a matter of fact, we have obtained a more general result
regarding the convexity of functions $\mathcal{F}(x_1,\ldots, x_n)$
of the form $\mathcal{F}(x_1,\ldots, x_n)=\prod_{i=1}^n f(x_i)$,
where $f(x)$ is a positive   function that has continuous second
derivative. This result may have applications in other areas of
mathematics.

 \vspace{1cm} \noindent \textbf{Acknowledgement}\;
The authors would like to thank Malaysian Academy of Sciences,
Ministry of Science, Technology  and Innovation for funding this
project under the Scientific Advancement Fund Allocation (SAGA) Ref.
No P96c.

\end{document}